\documentclass[article,aps,pra,11pt,tightenlines]{revtex4-2}
\usepackage{graphicx,epsfig}
\usepackage{amsmath}
\usepackage{amsfonts}
\usepackage{braket}
\usepackage{bbm}
\usepackage{dsfont}

\usepackage{physics}
\usepackage[dvipsnames]{xcolor}
\usepackage{algorithm2e} 
\usepackage{hyperref}
\RestyleAlgo{ruled}
\usepackage{amsthm}
\usepackage{overpic}

\makeatletter
\renewcommand{\p@subsection}{}
\renewcommand{\p@subsubsection}{}
\makeatother

\newcommand{\CPd}{\mathbb{CP}^{d-1}}
\DeclareMathOperator*{\argmin}{\arg\!\min}

\DeclareMathOperator{\EX}{\mathbb{E}}

\newcommand{\hil}{\mathcal{H}} 

\def\<{\langle}
\def\>{\rangle}

\newtheorem{theorem}{Theorem}
\newtheorem{lemma}[theorem]{Lemma}
\newtheorem{proposition}[theorem]{Proposition}

\newcommand{\REGRET}{\textnormal{Regret}}

\begin{document}


\title{\large Learning Pure Quantum States in Any Dimension (Almost) Without Regret}

\author{Josep Lumbreras$^{1,2}$} 
\email{josep.lz@ntu.edu.sg}
\author{Marco Tomamichel$^{1,3}$}

\affiliation{$^1$Centre for Quantum  Technologies, National University of Singapore, Singapore}
\affiliation{$^2$Nanyang Quantum Hub, School of Physical and Mathematical Sciences, Nanyang Technological University, Singapore}
\affiliation{$^3$Department of Electrical and Computer Engineering, National University of Singapore, Singapore}

\begin{abstract}
We extend quantum state tomography with minimal cumulative disturbance, first investigated in [arXiv:2406.18370], to arbitrary finite-dimensional pure states. A learner sequentially receives fresh copies of an unknown pure state, chooses a rank-one projector for each copy using the previous outcomes, and performs the corresponding two-outcome projective measurement. The goal is to learn the state while keeping the chosen projectors close to the unknown state in order to minimize disturbance. 
The qubit solution relies on the special geometry of the Bloch sphere and does not extend directly to qudits, where pure states form a curved manifold. We show that this obstruction can be overcome by working locally on the pure-state manifold. The algorithm proceeds in epochs. In each epoch, it fixes a current estimate, measures pairs of nearby rank-one projectors obtained by moving in opposite tangent directions, and takes differences of the corresponding outcomes. This gives an exact linear observation of the tangent component of the error.
The resulting local linear models are combined with a robust variance-adaptive estimator and a hot-start regularization that transfers precision across epochs. For every unknown pure state in dimension \(d\), after \(T\) measured copies, our protocol achieves cumulative regret \(\mathcal{O}(d^3\log^2 T)\), and at each intermediate time \(t\leq T\) its current estimate has online infidelity \(\mathcal{O}(d^3\log(T)/t)\). Hence, pure-state tomography with essentially no cumulative disturbance is not a peculiarity of qubits but a geometric phenomenon that persists for qudits.
\end{abstract}

\maketitle

\section{Introduction}

Quantum state tomography is usually formulated as the task of producing an accurate
classical description of an unknown quantum state from many independently prepared
copies. In the pure-state multi-armed quantum bandit problem (PSMAQB), introduced
in~\cite{lumbreras22bandit,lumbreras2024learning}, tomography is studied in a sequential setting
where the measurements themselves are part of the objective. At each round, the learner receives a fresh copy of an unknown pure state
$\rho=|\psi\rangle\!\langle\psi|$, chooses a rank-one projector
$A_t=|\phi_t\rangle\!\langle\phi_t|$, and performs the two-outcome
measurement $\{A_t,\mathbb I_d-A_t\}$. The expected reward is
$\operatorname{Tr}(\rho A_t)$, and the cumulative regret is
\begin{align}
    \REGRET(T)
    :=
    \sum_{t=1}^T
    \left(
        1-\operatorname{Tr}(\rho A_t)
    \right).
\end{align}
Equivalently, since both $\rho$ and $A_t$ are pure states,
\begin{align}
    1-\operatorname{Tr}(\rho A_t)
    =
    \frac{1}{2}\|\rho-A_t\|_F^2,
\end{align}
where $\|\cdot\|_F$ denotes the Frobenius norm.
Thus the regret measures how far, on average and over time, the chosen
measurements are from the unknown state. For rank-one projectors,
this quantity also controls the expected post-measurement disturbance of the
consumed copies. Indeed, if $q_t=\operatorname{Tr}(\rho A_t)$, the expected post-measurement
infidelity at round $t$ is $2q_t(1-q_t)$, while the regret term is
$1-q_t$; hence the cumulative expected disturbance is at most
$2\REGRET(T)$. The same additive objective also appears in state-agnostic
work extraction, where it has the thermodynamic interpretation of cumulative
dissipated energy caused by suboptimal control directions~\cite{lumbreras25dissipation}.

For qubits, the pure-state manifold is the Bloch sphere, and the local geometry can be
handled with two-dimensional Euclidean coordinates. For qudits, the pure-state manifold is
$\mathbb{CP}^{d-1}$, of real dimension $2(d-1)$, while the ambient space of trace-one
Hermitian matrices has dimension $d^2-1$. Thus the qudit problem is not obtained by
replacing the Bloch vector with a generalized Bloch vector. A naive linear-bandit
construction~\cite{lin1,lin2,lin3,pmlr-v247-lumbreras24a} in the ambient space introduces normal directions. These directions
are harmless in the qubit argument, but become a genuine obstruction in higher dimension.
Our main contribution is a new adaptive algorithm that works intrinsically with the geometry
of the pure-state manifold.

The main result of this work is that this geometric obstruction can be overcome.

\textbf{Main result (informal theorem).}
For every finite dimension $d\ge 2$, there is an adaptive protocol which learns any unknown $d$-dimensional qudit pure state
$\rho=|\psi\rangle \! \bra{\psi}$ using fresh copies of $\rho$ and two-outcome measurements defined by rank-one projectors. For any number of copies 
$T$, its expected cumulative regret satisfies
\begin{equation}
    \mathbb E\!\left[\mathrm{Regret}(T)\right]
    =
    \mathcal O\!\left(d^3\log^2 T\right).
\end{equation}
 Moreover, the algorithm outputs an online estimate at each round $t\in[T]$ that has infidelity
$\widetilde{\mathcal O}(d^3/t)$. Thus, for fixed dimension, online pure-state tomography can be
performed at the optimal $1/t$ infidelity while the cumulative disturbance grows only polylogarithmically in
the number of samples.

\textbf{Protocol overview.}
The algorithm proceeds in epochs. It first runs a short warm-up tomography routine to obtain
a constant-accuracy pure-state estimate $C_1$. At the beginning of each epoch $m$, the
algorithm fixes the current base state $C_m$ and works in the tangent space
$T_{C_m}\mathcal M$ of the pure-state manifold. Along an orthonormal tangent basis, it
constructs pairs of rank-one projectors by retracting in the positive and
negative tangent directions. Taking the difference of the two binary outcomes gives an
exact local linear observation of the tangent component of $\rho-C_m$. This removes the
curvature bias that would appear from using only one-sided measurements.

The tangent observations are combined through a variance-adaptive least-squares estimator.
The weights put more emphasis on directions with smaller variance, which is what allows the
regret analysis to exploit the decreasing statistical noise as we measure close to the unknown state. Since the observations are generated sequentially and the weights depend on past estimates, the estimator is made robust by a Median-of-Means construction. At the end
of the epoch, the estimated tangent displacement is retracted back to the pure-state
manifold, producing the next base state $C_{m+1}$. The statistical precision accumulated in
one epoch is transferred to the next through a hot-start regularization, rather than by
reusing tangent vectors from past tangent spaces.

Technically, the proof relies on three ingredients. First, symmetric retractions give an exact
linear model in the current tangent space. Second, hot-start regularization transfers
precision across changing tangent spaces. Third, a variance-adaptive Median-of-Means
estimator controls the sequential noise uniformly across epochs.

\textbf{Organization.}
The rest of the work is organized as follows. Section~\ref{sec:qudit_pre}
introduces the geometric ingredients needed in dimension $d>2$: the pure-state
manifold, its tangent spaces, the tangent projection, and the local retraction map.
Section~\ref{sec:linear_tangent_model} derives the exact tangent-space linear model
obtained from symmetric retracted measurements. Section~\ref{sec:action_selection_qudit} defines the measurement-selection rule used by the algorithm and analyzes the resulting design superoperator. The measurements are chosen along tangent directions determined by the current design superoperator, and the symmetric choice of directions keeps the design isotropic, so that its single eigenvalue grows at a controlled rate within each epoch. Section~\ref{sec:mom_qudit}
proves the variance-adaptive Median-of-Means estimate used inside each epoch.
Section~\ref{sec:qudit_algorithm} presents the full epoch-based algorithm, including the
warm-up step and the hot-start transfer of scalar precision across epochs.
The final part of Section~\ref{sec:qudit_algorithm} combines all the previous tools and results to prove 
the main regret and online infidelity guarantees. Finally, Section~\ref{sec:discussion} discusses open problem and potential applications of the framework.

\section{Qudit Generalization: Preliminaries}\label{sec:qudit_pre}

\subsection{Notation and conventions}

Let $[t]=\{1,\ldots,t\}$ for $t\in\mathbb{N}$. For real vectors
$x,y\in\mathbb{R}^m$ we write
$\langle x,y\rangle=\sum_{i=1}^m x_i y_i$ and denote the Euclidean norm by
$\|x\|_2$. If $A\in\mathbb{R}^{m\times m}$ is positive semidefinite, we write
$\|x\|_A^2=\langle x,Ax\rangle$. For a real symmetric matrix $A$, its
eigenvalues are ordered as
$\lambda_1(A)\leq \cdots \leq \lambda_m(A)$, with
$\lambda_{\min}(A)=\lambda_1(A)$ and
$\lambda_{\max}(A)=\lambda_m(A)$. For a random variable $X$, we denote its
expectation and variance by $\mathbb{E}[X]$ and $\operatorname{Var}(X)$,
respectively.

Throughout, $d$ denotes the Hilbert-space dimension and
$\hil=\mathbb{C}^d$. Let $\mathbb{H}_d$ be the real vector space of
$d\times d$ Hermitian matrices, and let
$\mathcal{S}_d=\{\rho\in\mathbb{H}_d:\rho\geq 0,\operatorname{Tr}(\rho)=1\}$
be the set of quantum states. We denote by
$\mathcal{S}_d^*=\{\rho\in\mathcal{S}_d:\rho^2=\rho\}$ the set of pure states,
or rank-one projectors. The identity matrix is denoted by $\mathbb{I}_d$, or
simply $\mathbb{I}$ when the dimension is clear. We equip $\mathbb H_d$ with the Frobenius, equivalently Hilbert--Schmidt,
inner product $
    \langle A,B\rangle := \operatorname{Tr}(AB)$ for $
     A,B\in\mathbb H_d$,
and write $\|A\|_F^2:=\langle A,A\rangle$, with associated Frobenius norm
$\|A\|_F=\sqrt{\operatorname{Tr}(A^2)}$. For two states
$\rho,\sigma\in\mathcal{S}_d$, the fidelity is
$F(\rho,\sigma)=\big(\operatorname{Tr}\sqrt{\sqrt{\sigma}\rho\sqrt{\sigma}}\big)^2$.
When convenient, we use Dirac notation and write a pure state as
$\rho=\ket{\psi}\!\bra{\psi}$.

\subsection{Problem setting and objective}\label{sec:problem_setting}

We consider the pure-state multi-armed quantum bandit (PSMAQB) problem in dimension $d$ introduced in~\cite{lumbreras22bandit,lumbreras2024learning}. An unknown environment state $\rho\in\mathcal S_d^*$
is fixed throughout the experiment. At each round $t\in[T_{\mathrm{total}}]$,
the learner receives a fresh copy of $\rho$, chooses adaptively a rank-one projector
$A_t\in\mathcal S_d^*$ based on previous outcomes and measurements (or actions), and performs the two-outcome measurement
$\{A_t,\mathbb I_d-A_t\}$. The observed outcome $X_t\in\{0,1\}$ satisfies
\begin{align}
    \mathbb E[X_t\mid A_t]=\operatorname{Tr}(\rho A_t).
\end{align}
We will write $\mathbb E_\rho$ for expectation with respect to the probability law
of the full adaptive algorithm when the unknown
environment state is $\rho$. This includes the randomness of the measurement
outcomes and any internal randomness of the algorithm.

The action $A=\rho$ has expected reward one (which is the highest one). Therefore the
instantaneous regret of playing $A_t$ is $
    1-\operatorname{Tr}(\rho A_t)$, and the cumulative regret is
\begin{align}
    \operatorname{Regret}(T_{\mathrm{total}})
    :=
    \sum_{t=1}^{T_{\mathrm{total}}}
    \left(1-\operatorname{Tr}(\rho A_t)\right).
\end{align}
Since both $\rho$ and $A_t$ are rank-one projectors,
\begin{align}
    1-\operatorname{Tr}(\rho A_t)
    =
    \frac12\|\rho-A_t\|_F^2.
\end{align}
Thus minimizing regret is equivalent to keeping the chosen measurements close to
the unknown state in squared Frobenius distance, or equivalently in pure-state
infidelity. 

\subsection{Dimensional mismatch and algorithmic indexing}

Having fixed the model and objective, we now explain why the qubit
linear-bandit reduction does not directly extend to qudits. In the qubit
case, this model admits a reduction to a classical linear bandit on the Bloch
sphere: writing
$\rho=(\mathbb{I}+\theta\cdot\sigma)/2$ and
$A=(\mathbb{I}+a\cdot\sigma)/2$, with $\theta,a\in\mathbb{S}^2$, the expected
reward is $(1+\langle\theta,a\rangle)/2$. The corresponding qubit algorithm
and regret analysis were developed in~\cite{pmlr-v247-lumbreras24a,lumbreras2024learning}. The goal
of this section is to explain why this Bloch-sphere argument does not
extend directly to qudits, and how it can be replaced by a tangent-space
construction on the pure-state manifold.

For arbitrary qudits ($d>2$), however, the qubit geometry no longer generalizes directly.  The pure states form the real $(2d-2)$-dimensional manifold
\begin{align}
    \mathcal M=\mathcal S_d^*\cong \mathbb{CP}^{d-1} ,
\end{align}
inside the affine space of trace-one Hermitian matrices, whose real dimension is $d^2-1$.  Thus, at a pure state $C\in\mathcal M$, the Frobenius-orthogonal complement of the tangent space $T_C\mathcal M$ inside this affine space has dimension $(d^2-1)-(2d-2)=(d-1)^2$.
Equivalently, since all pure states also have fixed Frobenius norm, one may remove the single radial direction $C-\mathbb{I}_d/d$.  Inside the corresponding generalized Bloch sphere, the pure-state manifold has $d^2-2d$ additional normal directions.  These directions are absent for qubits, because $\mathbb{CP}^1$ is the full Bloch sphere.

Consequently, a classical linear bandit method applied directly in the ambient $d^2-1$-dimensional parameter space would try to estimate directions that are normal to the pure-state manifold.  But near the optimal action, admissible rank-one projectors can move only along $T_C\mathcal M$ to first order; their components in the extra normal directions appear only through the curvature of the manifold, hence at second order.  The leading-order information relevant for reducing regret is therefore tangent information, not information in the $(d-1)^2$ ambient-normal directions, or equivalently the $d^2-2d$ non-radial normal directions inside the Bloch sphere.

Another approach would be to formulate the problem directly with unit complex vectors, since the action can be parameterized as $A = \ket{\psi_a} \! \bra{\psi_a} \in \mathcal{S}^*_d$. However, the expected outcome of the measurements would then scale as $|\braket{\psi}{\psi_a}|^2$, which is quadratic with respect to the measurement direction $\ket{\psi_a}$. This means we would lose the desirable analytical properties of the linear least-squares estimator, which possesses a convenient closed form when inverted and provides clean concentration bounds for adaptive measurements. 

We resolve this dimensional mismatch by formulating the algorithm intrinsically on
$\mathbb{CP}^{d-1}$, rather than in the ambient Hermitian space. By operating directly on the manifold, we restrict the learning process to explore only the directions allowed by the pure-state constraint. In this section, we provide the formal geometric framework for an epoch-based algorithm that, given an online estimate, performs measurements exclusively within the local tangent space of that estimate. In this way, we will see that the estimated quantity remains linear and matches the degrees of freedom of the pure state manifold.  We will use standard differential-geometric notions on this manifold, such as tangent spaces and local retractions; see, e.g.,~\cite{bengtsson2017geometry,absil2008optimization}


The main difference from the qubit setting of~\cite{lumbreras2024learning}
is the treatment of curvature. A direct transcription of the qubit algorithm
to $\mathbb{CP}^{d-1}$, based on a naive local Euclidean approximation, would
introduce residual terms coming from the normal directions of the manifold.
These residuals would accumulate into a linear contribution to the regret. To
avoid this, we introduce an \textit{epoch-based} algorithm that temporarily
freezes the local coordinate system determined by the current estimator.

The execution is structured by the following hierarchy of indices. We denote the real tangent dimension by
\begin{equation}\label{eq:dtan_def_qudit}
    d_{\mathrm{tan}}
    :=
    \dim_{\mathbb{R}} T_{C_m}\mathcal{M}
    =
    2(d-1),
\end{equation}
which is independent of the epoch $m$. The global physical measurement budget is denoted by $T_{\mathrm{total}}$. This symbol should not be confused with $T_m$: the quantity $T_m$ is the number of tangent-basis update steps inside epoch $m$, while $T_{\mathrm{total}}$ counts the total number of measurements used by the full algorithm, including a warm-up cost $T_0$ that we will define later. We explain below the different parts of the algorithms and the indices we will use to label them. The exact constructions of each part will be presented later, but for clarity we already state the hierarchy of indices that we will use.

\begin{itemize}
\item The adaptive phase is divided into $M$ \emph{total epochs}, indexed by $m \in [M]$. At the start of each epoch $m$, we fix a rank-1 \emph{base state} $C_m \in \mathcal{S}_d^*$, which serves as our current linear estimate of the unknown state $\ket{\psi}$. By freezing $C_m$, we define a stable local coordinate system in which to compute our measurement directions.

\item Each epoch consists of $T_m$ \emph{steps}, indexed by $s\in [T_m]$. At each step $s$, we define measurement directions across all $d_{\mathrm{tan}}$ tangent directions specified by the local coordinate system of the base state $C_m$. The running index for the tangent directions is $i \in [d_{\mathrm{tan}}]$.

\item As in the qubit construction of~\cite{lumbreras2024learning}, we use repeated
measurements and a Median-of-Means (MoM) aggregation to obtain robust
confidence bounds under bounded-variance noise. Every measurement is sampled
$N = 2 \lceil 12 \log(T_{\mathrm{total}}/\delta) \rceil$ times, and the
repetition index is $j \in [N]$.

\item We also use the symmetric-pair idea from~\cite{pmlr-v247-lumbreras24a,lumbreras2024learning}, but now the two measurement directions are generated by
moving in opposite tangent directions and retracting back to
$\mathcal{S}_d^*$. Specifically, for every epoch $m$, step $s$, and tangent
direction $i$, we select two symmetric rank-one projectors
$A_{m,s,i}^+ \in \mathcal{S}_d^*$ and
$A_{m,s,i}^- \in \mathcal{S}_d^*$. This symmetric construction
cancels the nonlinear normal curvature components of the manifold, yielding a
linear model. 
\end{itemize}

Thus, given our unknown state
$\rho = \ket{\psi}\!\bra{\psi}$, the regret for the qudit algorithm can be
decomposed as
\begin{equation} \label{eq:regret_def}
  \mathrm{Regret}(T_{\mathrm{total}}) = \frac{1}{2}  \sum_{m=1}^{M} \sum_{s=1}^{T_m} \sum_{j=1}^N \sum_{i=1}^{d_{\mathrm{tan}}} \Big( \|\rho - A_{m,s,i}^+\|_F^2 + \|\rho - A_{m,s,i}^-\|_F^2 \Big),
\end{equation}
where we used $\operatorname{Tr}(\rho^2) = \operatorname{Tr}(A^2) = 1$ with
\begin{equation}
    \|\rho - A\|_F^2 = \operatorname{Tr}\big((\rho - A)^2\big) = \operatorname{Tr}(\rho^2) + \operatorname{Tr}(A^2) - 2\operatorname{Tr}(\rho A) = 2\big(1 - \operatorname{Tr}(\rho A)\big).
\end{equation}
Therefore, minimizing the infidelity is mathematically identical to minimizing half the squared Frobenius distance: $1 - \operatorname{Tr}(\rho A) = \frac{1}{2}\|\rho - A\|_F^2$. Then, for a run containing $M$ complete epochs,
\begin{align}
    T_{\mathrm{total}}
    =
    T_0
    +
    2N d_{\mathrm{tan}}
    \sum_{m=1}^{M} T_m,
\end{align}
up to truncating the final epoch if the budget ends in the middle of an epoch. The factor $2$ comes from the symmetric pair of actions $A^+_{m,s,i},A^-_{m,s,i}$, the factor $d_{\mathrm{tan}}$ from the tangent directions, and the factor $N$ from the MoM repetitions.

The full algorithm is stated in Section~\ref{sec:qudit_algorithm}. Before
that, we introduce the geometric and statistical ingredients needed to define
the tangent-space measurements, the local linear estimator, and the
epoch-to-epoch hot start.

\subsection{Geometry of the Pure State Manifold}

In this section, we now describe the tangent spaces and retractions used by the algorithm. Having established the need to avoid the ambient dimensions, our approach relies on constructing local linear models explicitly on the curved manifold.

We achieve this by fixing a \emph{base state} $C_m \in \mathcal{S}_d^*$ at the onset of each epoch $m$. The base state $C_m=|\psi_{c_m}\rangle\langle\psi_{c_m}|$ is the current
pure-state estimate of $\rho$. During epoch $m$, all local coordinates are
taken in the local tangent space $T_{C_m}\mathcal M$, defined by
\begin{equation} \label{eq:qudit_tangent}
    T_{C_m}\mathcal{M} := \left\{ V = \frac{1}{\sqrt{2}}\big( \dyad{\phi}{\psi_{c_m}} + \dyad{\psi_{c_m}}{\phi} \big) \;\bigg|\;
\ket{\phi} \in \mathbb{C}^d, \, \braket{\phi}{\psi_{c_m}} = 0 \right\}.
\end{equation}
To formalize the structure of the local tangent space introduced above, we must establish that it indeed forms a valid real vector space of Hermitian matrices. 

\begin{lemma}
Every element $V \in T_{C_m}\mathcal{M}$ is a Hermitian matrix and $T_{C_m}\mathcal{M}$ is a real vector space.
\end{lemma}

\begin{proof}
The first statement can be verified by inspection.
Let $V_1, V_2 \in T_{C_m}\mathcal{M}$ and $c \in \mathbb{R}$. By definition, there exist vectors $|\phi_1\rangle, |\phi_2\rangle \in \mathbb{C}^d$ such that $\langle\phi_1|\psi_{c_m}\rangle = \langle\phi_2|\psi_{c_m}\rangle = 0$. The sum of these operators is given by
\begin{equation}
V_1 + V_2 = \frac{1}{\sqrt{2}} \left( (|\phi_1\rangle + |\phi_2\rangle)\langle\psi_{c_m}| + |\psi_{c_m}\rangle(\langle\phi_1| + \langle\phi_2|) \right).
\end{equation}
Defining $|\phi'\rangle = |\phi_1\rangle + |\phi_2\rangle$, we see that $\langle\phi'|\psi_{c_m}\rangle = 0$, meaning $V_1 + V_2 \in T_{C_m}\mathcal{M}$. Furthermore, scalar multiplication by $c$ yields
\begin{equation}
c V_1 = \frac{1}{\sqrt{2}} \left( (|c \phi_1\rangle)\langle\psi_{c_m}| + |\psi_{c_m}\rangle\langle c\phi_1| \right),
\end{equation}
where we have used the fact that $c$ is real. Letting $|\phi''\rangle = |c\phi_1\rangle$, we find $\langle\phi''|\psi_{c_m}\rangle = 0$. Since $T_{C_m}\mathcal{M}$ is closed under matrix addition and real scalar multiplication, it forms a real vector space. 
\end{proof}

With the geometric structure of the pure state manifold well-posed, we will be able later to generate measurements by selecting tangent generators $V \in T_{C_m}\mathcal{M}$.

To construct our exact local linear model, we must ensure the learning process relies exclusively on the degrees of freedom of the pure state manifold. We need a mechanism to filter out the normal dimensions of the ambient space and isolate these valid tangent variations. We formalize this mechanism via the following lemma, which establishes the precise orthogonal projector onto the local tangent space.

\begin{lemma}\label{lem:projector_properties}
For any base state $C_m = |\psi_{c_m}\rangle\langle\psi_{c_m}| \in \mathcal{S}_d^*$, define the orthogonal projection map $\mathcal{P}_{T_{C_m}} : \mathbb{H}_d \rightarrow \mathbb{H}_d$ as
\begin{equation}
\mathcal{P}_{T_{C_m}}(X) := C_m X (\mathbb{I}-C_m) + (\mathbb{I}-C_m) X C_m. \label{eq:projector_def}
\end{equation}
Then, $\mathcal{P}_{T_{C_m}}$ is an orthogonal projector whose image is the tangent space $T_{C_m}\mathcal{M}$.
\end{lemma}

\begin{proof}
To prove that $\mathcal{P}_{T_{C_m}}$ is a projector onto $T_{C_m}\mathcal{M}$, we must show that its image lies in $T_{C_m}\mathcal{M}$, that it is idempotent ($\mathcal{P}_{T_{C_m}}^2 = \mathcal{P}_{T_{C_m}}$), and that it acts as the identity on any element already in $T_{C_m}\mathcal{M}$.

First, let $X \in \mathbb{H}_d$ be an arbitrary Hermitian matrix. We can define an unnormalized vector $|\tilde{\phi}\rangle = \sqrt{2}(\mathbb{I}-C_m)X|\psi_{c_m}\rangle$. Because $C_m = |\psi_{c_m}\rangle\langle\psi_{c_m}|$, it follows that $(\mathbb{I}-C_m)|\psi_{c_m}\rangle = 0$, which ensures the orthogonality condition $\langle \psi_{c_m} | \tilde{\phi} \rangle = 0$. We can rewrite the action of the superoperator as
\begin{align}
\mathcal{P}_{T_{C_m}}(X) &= |\psi_{c_m}\rangle\langle\psi_{c_m}| X (\mathbb{I}-C_m) + (\mathbb{I}-C_m) X |\psi_{c_m}\rangle\langle\psi_{c_m}| \nonumber \\
&= \frac{1}{\sqrt{2}} \left( |\psi_{c_m}\rangle \langle \tilde{\phi}| + |\tilde{\phi}\rangle \langle \psi_{c_m} | \right).
\end{align}
This matches the definition of the tangent space~\eqref{eq:qudit_tangent} (with $|\phi\rangle = |\tilde{\phi}\rangle$). Thus, the image of $\mathcal{P}_{T_{C_m}}$ is contained in $T_{C_m}\mathcal{M}$.

Second, we check idempotence. Using the projective properties $C_m^2 = C_m$, $(\mathbb{I}-C_m)^2 = \mathbb{I}-C_m$, and the orthogonality $C_m(\mathbb{I}-C_m) = 0$, we apply the superoperator twice
\begin{align}
\mathcal{P}_{T_{C_m}}(\mathcal{P}_{T_{C_m}}(X)) &= C_m \left[ C_m X (\mathbb{I}-C_m) + (\mathbb{I}-C_m) X C_m \right] (\mathbb{I}-C_m) \nonumber \\
&\quad + (\mathbb{I}-C_m) \left[ C_m X (\mathbb{I}-C_m) + (\mathbb{I}-C_m) X C_m \right] C_m \nonumber \\
&= C_m X (\mathbb{I}-C_m) + (\mathbb{I}-C_m) X C_m \nonumber \\
&= \mathcal{P}_{T_{C_m}}(X).
\end{align}
Thus, $\mathcal{P}_{T_{C_m}}$ is idempotent. 

Finally, consider an arbitrary element $V \in T_{C_m}\mathcal{M}$. By definition, $V = \frac{1}{\sqrt{2}}(|\phi\rangle\langle\psi_{c_m}| + |\psi_{c_m}\rangle\langle\phi|)$ with $\langle\phi|\psi_{c_m}\rangle = 0$. Since $C_m|\psi_{c_m}\rangle = |\psi_{c_m}\rangle$ and $C_m|\phi\rangle = 0$, a direct substitution yields $\mathcal{P}_{T_{C_m}}(V) = V$. Therefore, the superoperator is surjective onto $T_{C_m}\mathcal{M}$, confirming it is the exact orthogonal projector onto the local tangent space.

It remains to check that the projection is orthogonal with respect to the Frobenius
inner product. Let $Q:=I-C_m$. Since $C_m=C_m^\dagger=C_m^2$ and $Q=Q^\dagger=Q^2$, for
Hermitian matrices $X,Y$ we have
\begin{align}
\begin{aligned}
\left\langle Y,\mathcal P_{T_{C_m}}(X)\right\rangle
&=
\operatorname{Tr}\!\left[
Y\left(C_mXQ+QXC_m\right)
\right]  \\
&=
\operatorname{Tr}(YC_mXQ)+\operatorname{Tr}(YQXC_m) \\
&=
\operatorname{Tr}(QYC_mX)+\operatorname{Tr}(C_mYQX) \\
&=
\operatorname{Tr}\!\left[
\left(C_mYQ+QYC_m\right)X
\right] \\
&=
\left\langle \mathcal P_{T_{C_m}}(Y),X\right\rangle .
\end{aligned}
\end{align}
Thus $\mathcal P_{T_{C_m}}$ is self-adjoint. Since it is also idempotent and has image
$T_{C_m}\mathcal M$, it is the Frobenius-orthogonal projector onto $T_{C_m}\mathcal M$.

\end{proof}

Having established that $\mathcal{P}_{T_{C_m}}$ isolates the valid degrees of freedom, how do we map these directions back to the curved pure state manifold $\mathcal{M}$? 

For qubits, the Bloch-sphere geometry used in~\cite{lumbreras2024learning}
allows one to return to the action set by a simple vector renormalization,
$a/\|a\|_2$. For qudits, we use the geodesic in the two-dimensional subspace generated by
$|\psi_{c_m}\rangle$ and the tangent direction.

The geometry of the pure state space guarantees that any geodesic connecting the base state $\ket{\psi_{c_m}}$ and an orthogonal direction $\ket{\phi}$ is entirely confined to the 2D subspace spanned by them. We can construct a geodesic between two orthogonal pure states by just connecting them with a ``circle" as
\begin{equation}
    \ket{\gamma_V(\pm\tau)} = \cos\left(\frac{\tau}{\sqrt{2}}\right)\ket{\psi_{c_m}} \pm \sin\left(\frac{\tau}{\sqrt{2}}\right)\ket{\phi},
\end{equation}
where the $\sqrt{2}$ factors arise from the Frobenius normalization $\operatorname{Tr}(V^2) = 1$. To translate this state-vector evolution into the density matrix formalism and ensure our actions remain valid rank-1 projectors, we formalize the mapping via the following lemma.

\begin{lemma}
For any base state $C_m = \ket{\psi_{c_m}} \! \bra{\psi_{c_m}} \in \mathcal{S}_d^*$ and a normalized tangent vector $V \in T_{C_m}\mathcal{M}$, $\| V \|_F = 1$ characterized by the orthogonal state $\ket{\phi}$, $V = \frac{1}{\sqrt{2}}(\ket{\phi}\!\bra{\psi_{c_m}} + \ket{\psi_{c_m}}\!\bra{\phi} )$, the exact mapping back to the pure state manifold along a step size $\pm\tau \in \mathbb{R}$ is given by the Retraction operator defined as
\begin{equation} \label{eq:qudit_retract}
    \operatorname{Retract}_{C_m}(\pm \tau V) := \cos^2\left(\frac{\tau}{\sqrt{2}}\right) C_m + \sin^2\left(\frac{\tau}{\sqrt{2}}\right) \dyad{\phi} \pm \frac{1}{\sqrt{2}}\sin(\sqrt{2}\tau) V.
\end{equation}
\end{lemma}

\begin{proof}
The exact update is obtained by expanding the projector $\ket{\gamma_V(\pm\tau)} \! \bra{\gamma_V(\pm\tau)}$ using the state-vector geodesic defined above, and recalling that $\bra{\psi_m}\phi \rangle = 0$. Since~\eqref{eq:qudit_retract} is equal to $\ket{\gamma_V(\pm\tau)} \! \bra{\gamma_V(\pm\tau)}$ then it is a valid rank-1 projector or pure state.
\end{proof}

\section{Linear Model in the Tangent Space}\label{sec:linear_tangent_model}

Having established the necessary tools to project states into local tangent spaces and retract them back to the manifold, we are now ready to define the linear estimation model for our algorithm. We begin by constructing a measurement scheme such that the expected outcomes remain linear with respect to both the unknown parameter and the measurement directions.

We first outline how the algorithm selects measurement directions. Throughout this section, we fix an epoch $m\in [M]$ and its corresponding base state $C_m\in\mathcal{S}_d^*$. At step $s$ of epoch $m$, the algorithm selects a tangent vector $V_s \in T_{C_m}\mathcal{M}$ and determines a geodesic step size $\tau_s$. The corresponding physical measurement direction is generated via the retraction as $A_s = \operatorname{Retract}_{C_m} (\tau_s V_s)$. Because the unknown environment is a pure state $\rho = \ket{\psi} \! \bra{\psi} \in \mathcal{S}^*_d$, the observed binary reward/outcome $X_s$, governed by Born's rule, is given by
\begin{align}
    X_s &= \langle \rho , A_s \rangle + \eta_s \nonumber \\
        &= \cos^2\left(\frac{\tau_s}{\sqrt{2}}\right) \langle \rho, C_m \rangle + \sin^2\left(\frac{\tau_s}{\sqrt{2}}\right) \langle \rho, \dyad{\phi} \rangle + \frac{1}{\sqrt{2}}\sin(\sqrt{2}\tau_s) \langle \rho, V_s \rangle + \eta_s,
\end{align}
where $\eta_s$ represents the statistical noise. 

This expected reward highlights the difficulty of performing pure-state tomography as a linear bandit problem. If we attempt to parameterize the learning using the underlying quantum state vector $\ket{\gamma(\tau_s)}$, the Born rule yields a probability that is quadratic, breaking the linear assumptions. Conversely, lifting the parameterization to the full density matrix space using the action $A_s$ makes the reward linear via the Frobenius inner product $\langle \rho, A_s \rangle$. However, as established in the previous section, the ambient density matrix space possesses $d^2-1$ dimensions, while our physical measurement directions are confined to the $2(d-1)$-dimensional pure state manifold $\CPd$.  

To formulate an exact linear model within the local tangent space, we must filter out the manifold's curvature and isolate the only tangent relevant part. We achieve this by defining our exact target parameter as the projection of the true state $\rho$ onto the tangent space of our current base state $C_m \in \mathcal{S}^*_d$ as
\begin{equation}\label{eq:tangent_parameter}
    \Delta^{(m)}_* := \mathcal{P}_{T_{C_m}}(\rho) = \mathcal{P}_{T_{C_m}}(\rho - C_m) \in T_{C_m}\mathcal{M},
\end{equation}
where the second equality holds trivially since $\mathcal{P}_{T_{C_m}}(C_m) = 0$. 

To extract information about this target parameter, we will select a symmetric pair of measurements/actions $A_s^\pm = \operatorname{Retract}_{C_m}(\pm \tau_s V_s)$ centered around the base state $C_m$. Expanding the retraction operator using~\eqref{eq:qudit_retract}, we can group the terms into even and odd parities with respect to the step size as
\begin{align}
    A_s^\pm = M_s^{\text{even}} \pm O_s,
\end{align}
where the even term
\begin{align}
    M_s^{\text{even}} := \cos^2\left(\frac{\tau_s}{\sqrt{2}}\right) C_m + \sin^2\left(\frac{\tau_s}{\sqrt{2}}\right) \dyad{\phi},
\end{align}
contains the base state and the curvature components, while the odd term
\begin{align}\label{eq:odd_term}
    O_s := \frac{1}{\sqrt{2}} \sin(\sqrt{2}\tau_s) V_s ,
\end{align}
contains the pure tangent variation. We observe independent binary rewards $X_s^\pm = \langle \rho , A_s^\pm \rangle + \eta_s^\pm$. By defining the difference-reward as $Y_s = \frac{1}{2}(X_s^+ - X_s^-)$, the normal curvature term $M_s^{\text{even}}$ cancels out, yielding
\begin{equation}
    Y_s = \langle \rho , O_s \rangle + \epsilon_s, \quad \text{where } \epsilon_s := \frac{\eta_s^+ - \eta_s^-}{2}.
\end{equation}
We can express this expectation in terms of our tangent target parameter $\Delta^{(m)}_*$ by injecting the projection superoperator $\mathcal{P}_{T_{C_m}}$. Using the algebraic identity $O_s = \mathcal{P}_{T_{C_m}}(O_s) + \big(O_s - \mathcal{P}_{T_{C_m}}(O_s)\big)$, we obtain
\begin{equation} \label{eq:projected_reward}
    \langle \rho , O_s \rangle = \langle \mathcal{P}_{T_{C_m}}(\rho), \mathcal{P}_{T_{C_m}}(O_s) \rangle + \langle \rho , O_s - \mathcal{P}_{T_{C_m}}(O_s) \rangle.
\end{equation}
This decomposition reveals the necessity of our epoch-based algorithmic design. If the base state $C_m$ were updated continuously at every step $s$, the tangent generators $O_s$ constructed at past steps would not align with the current tangent space. Applying a new projection $\mathcal{P}_{T_{C_t}}$ to past measurements would generate a non-zero residual term $r_s = \langle \rho , O_s - \mathcal{P}_{T_{C_t}}(O_s) \rangle \neq 0$. This residual would accumulate an $\mathcal{O}(T)$ bias into the estimator.

By fixing the base state $C_m$ for an entire epoch $m$, the vectors $V_s$ (and consequently $O_s$) are constructed \textit{strictly within} $T_{C_m}\mathcal{M}$. Because projecting an element that already resides in the target subspace acts as the identity map, we have $\mathcal{P}_{T_{C_m}}(O_s) = O_s$. Therefore, the residual drift is zero throughout the epoch. 

This ``no residual drift'' claim remains valid  because old measurement outcomes and old tangent vectors are discarded at the start of each new epoch. As we will explain later, to prevent the mixing of tangent vectors defined for different base states $C_m$, we must introduce a ``hot-start mechanism'' to carry forward the statistical estimation gathered at the end of each epoch.

Furthermore, we observe that
\begin{align}
    \langle \Delta^{(m)}_*, O_s \rangle = \langle \rho - C_m, O_s \rangle = \langle \rho , O_s \rangle - \langle C_m, O_s \rangle.
\end{align}
Because $O_s \in T_{C_m}\mathcal{M}$ resides purely in the tangent space of $C_m$, it is orthogonal to the base state, yielding $\langle C_m, O_s \rangle = 0$. Substituting this into~\eqref{eq:projected_reward} yields our \emph{linear model within the local tangent space}
\begin{equation} \label{eq:exact_linear}
    Y_s = \langle \Delta^{(m)}_*, O_s \rangle + \epsilon_s.
\end{equation}
Thus, the primary quantity of interest to estimate is $\Delta^{(m)}_*$, as it is linear with respect to the engineered observation vector $O_s$~\eqref{eq:odd_term}. 

However, when transitioning to a new epoch $m+1$ with a freshly updated base
state $C_{m+1}$, we cannot naively reuse the vectors gathered in
$T_{C_m}\mathcal{M}$. The new target parameter $\Delta_*^{(m+1)}$ generally
differs from $\Delta_*^{(m)}$, because these are projections of the unknown
state $\rho$ onto tangent spaces based at different points. This is one of the
main differences from the qubit algorithm of~\cite{lumbreras2024learning},
where the Bloch-sphere parametrization provides a fixed ambient coordinate
system. In the qudit setting, we therefore need a mechanism for transferring
statistical precision across epochs without transporting old tangent vectors.

\subsection{Tangent-Space Least-Squares Estimator and Design Superoperator}

With the exact linear model established in the previous section, our next goal is to estimate the target parameter $\Delta^{(m)}_*$ from a history of observed rewards (measurement outcomes) and actions (measurement directions). 

Following the weighted least-squares structure used in the qubit analysis
of~\cite{lumbreras2024learning}, we use a collection of linear observations
$(Y_l,O_l)_{l=1}^s$ gathered during epoch $m\in[M]$, together with
variance-adaptive weights $\{\omega_l\}_{l=1}^s$ that reflect the
vanishing-variance structure of the measurement noise. Since the unknown local
quantity is the tangent vector $\Delta\in T_{C_m}\mathcal{M}$, we define the
regularized least-squares loss
\begin{align}
    \mathcal{L}(\Delta) = \frac{\mu_{m-1}}{2} \langle \Delta , \Delta \rangle + \frac{1}{2} \sum_{l=1}^s \omega_l \big( Y_l - \langle \Delta, O_l \rangle\big)^2,
\end{align}
where $\mu_{m-1} > 0$ is a regularization constant. Why do we vary this constant at each epoch? As we will formally show later, updating $\mu_{m-1}$ provides a ``hot-start'' mechanism, allowing the estimator to inherit the estimation precision achieved in the previous epoch without mixing incompatible tangent spaces. 

To find the minimum that defines our least-squares estimator, we take the Fr\'echet derivative with respect to $\Delta \in T_{C_m}\mathcal{M}$ and set it to zero as
\begin{equation}
    \mu_{m-1} \Delta - \sum_{l=1}^s \omega_l \big( Y_l - \langle \Delta, O_l \rangle \big) O_l = 0 \implies \mu_{m-1} \Delta + \sum_{l=1}^s \omega_l O_l \langle O_l, \Delta \rangle = \sum_{l=1}^s \omega_l Y_l O_l.
\end{equation}
To isolate $\Delta$, we view $T_{C_m}\mathcal M$ as a real Hilbert space
equipped with the Frobenius inner product. 
The above equation can then be written as
\begin{equation}
     \mathcal V_s^{\mathrm{tan}}(\Delta)
     =
     \sum_{l=1}^s \omega_l Y_l O_l,
\end{equation}
where $\mathcal V_s^{\mathrm{tan}}:T_{C_m}\mathcal M\to T_{C_m}\mathcal M$
is the tangent design superoperator defined by
\begin{equation}\label{eq:tangent_design_matrix}
    \mathcal V_s^{\mathrm{tan}}(X)
    :=
    \mu_{m-1} X
    +
    \sum_{l=1}^s
    \omega_l
    \langle O_l,X\rangle O_l,
    \qquad
    X\in T_{C_m}\mathcal M .
\end{equation}
Since $\mu_{m-1}>0$, the operator $\mathcal V_s^{\mathrm{tan}}$ is strictly
positive on the finite-dimensional real Hilbert space $T_{C_m}\mathcal M$ which means that
for every nonzero $X\in T_{C_m}\mathcal M$,
\begin{equation}
    \left\langle X,\mathcal V_s^{\mathrm{tan}}(X)\right\rangle
    =
    \mu_{m-1}\|X\|^2
    +
    \sum_{l=1}^s
    \omega_l
    \langle O_l,X\rangle^2
    >
    0 .  \label{eq:norm-of-X}
\end{equation}
Thus, $\mathcal V_s^{\mathrm{tan}}$ is invertible as a linear map on $T_{C_m}\mathcal M$. Equivalently, in any Frobenius-orthonormal eigenbasis $\{e_r\}_{r=1}^{d_{\mathrm{tan}}}$ satisfying $\mathcal V_s^{\mathrm{tan}}e_r=\lambda_r e_r$, with $\lambda_r>0$, the inverse acts diagonally as $(\mathcal V_s^{\mathrm{tan}})^{-1}e_r=\lambda_r^{-1}e_r$.
Recall that positivity, eigenvectors and eigenvalues are always understood with
respect to the Frobenius inner product on the tangent space. 
The least-squares estimator is
therefore
\begin{equation} \label{eq:tangent_LSE}
    \widehat{\Delta}_s
    :=
    \left(\mathcal V_s^{\mathrm{tan}}\right)^{-1} \left(
    \sum_{l=1}^s \omega_l Y_l O_l \right) .
\end{equation}
This gives the tangent-space linear estimator for $\Delta_*^{(m)}$.

\subsection{About tangent-space superoperators}
We first state the tangent-space operator conventions that we will use. For a fixed
base state $C_m$, every tangent-space superoperator is a linear map
$\mathcal A:T_{C_m}\mathcal M\to T_{C_m}\mathcal M$ acting on Hermitian tangent
matrices. Self-adjointness, positivity, eigenvectors, and eigenvalues are always
understood with respect to the real Hilbert space
$(T_{C_m}\mathcal M,\langle\cdot,\cdot\rangle)$. Thus, if $\mathcal A$ is
self-adjoint, a Frobenius-orthonormal eigenbasis satisfies
\begin{align}
\mathcal A (v_i)=\lambda_i v_i,
\qquad
\langle v_i,v_j\rangle=\delta_{ij}, \quad v_i \in T_{C_m}\mathcal{M}.
\end{align}
If $\mathcal A$ is positive definite, all $\lambda_i>0$. We write
$\lambda_{\min}(\mathcal A)$ and $\lambda_{\max}(\mathcal A)$ for
the smallest and largest tangent-space eigenvalues.

For a positive superoperator $\mathcal A:T_{C_m}\to T_{C_m}$, define
\begin{align}
\|X\|_{\mathcal A}^2
:=
\langle X,\mathcal A ( X) \rangle,
\qquad X\in T_{C_m}.
\end{align}
For two self-adjoint tangent-space operators $\mathcal A,\mathcal B$,
we write
\begin{align}
\mathcal A\preceq\mathcal B , \quad
\text{if and only if} \quad
\langle H,\mathcal A ( H)\rangle
\le
\langle H,\mathcal B (H )\rangle
\end{align}
for every $H\in T_{C_m}$. 
Finally we denote by
\begin{align}
\mathcal I_{T_{C_m}\mathcal M}:T_{C_m}\mathcal M\to T_{C_m}\mathcal M, \quad \mathcal I_{T_{C_m}\mathcal M}(X)=X \quad \text{ for every } X\in T_{C_m}\mathcal M,
\end{align}
the identity superoperator on the tangent space.

\section{Measurement Selection, hot start and Eigenvalue Control}\label{sec:action_selection_qudit}

In this section we define the measurement selection rule that we will use in our algorithm. To describe this, we start by fixing the epoch $m\in [M]$ with its base state $C_m$.

First note that at step $s$ of epoch $m$, the design superoperator $\mathcal V_{s-1}^{\text{tan}}:T_{C_m}\mathcal{M}\rightarrow T_{C_m}\mathcal{M}$ is a self-adjoint positive 
tangent-space superoperator. By the spectral theorem, it
admits a Frobenius-orthonormal eigenbasis
$\{v_{s,i}\}_{i=1}^{d_{\mathrm{tan}}}$ contained in $T_{C_m}\mathcal M$, so these
vectors are valid inputs for the retraction operator~\eqref{eq:qudit_retract}.
Letting $\lambda_{s-1}:=\lambda_{\min}(\mathcal V_{s-1}^{\mathrm{tan}})$, we
explore the manifold by generating $d_{\mathrm{tan}}$ pairs of symmetric
measurement directions using the adaptive step size
$\tau_s=1/\sqrt{\lambda_{s-1}}$, namely
\begin{equation}
    A_{s,i}^{\pm} = \operatorname{Retract}_{C_m}\left(\pm\frac{1}{\sqrt{\lambda_{s-1}}}v_{s,i}\right). \label{eq:action_from_v}
\end{equation}
The above are valid measurement directions since $A_{s,i}^{\pm}\in\mathcal{S}_d^*$. However, recall from~\eqref{eq:exact_linear} that for the estimation step, the classical data we input into our tangent linear least-squares estimator~\eqref{eq:tangent_LSE} is the projected component $ \mathcal{P}_{T_{C_m}}(A_{s,i}^{\pm})$.

The following lemma establishes how to project these physical actions back onto the tangent space.

\begin{lemma}\label{lem:tangent_annihilation}
    Let $C_m = |\psi_m\rangle \! \langle \psi_m |\in\mathcal{S}^*_d$. Let $v_{s,i} \in T_{C_m}\mathcal{M}$ be any normalized eigenvector of $\mathcal{V}^{\mathrm{tan}}_{s-1}$ such that $v_{s,i} = \frac{1}{\sqrt{2}}( \dyad{\phi}{\psi_{m}} + \dyad{\psi_{m}}{\phi} )$ for some $|\phi\rangle\in\mathbb{C}^d$ such that $\langle \phi |\psi_m \rangle = 0$. Then the orthogonal projection of the corresponding actions $A_{s,i}^{\pm}$ as defined in~\eqref{eq:action_from_v} onto the tangent space yields
    \begin{equation}
        \mathcal{P}_{T_{C_m}}(A_{s,i}^{\pm}) = \pm \frac{1}{\sqrt{2}}\sin \left(\sqrt{\frac{2}{{\lambda_{s-1}}}}\right)v_{s,i}.
    \end{equation}
\end{lemma}

\begin{proof}
    Recall from~\eqref{eq:qudit_retract} that the generated physical action is given by
    \begin{equation*}
        A_{s,i}^{\pm} = \cos^2\left(\frac{\tau_s}{\sqrt{2}}\right) C_m + \sin^2\left(\frac{\tau_s}{\sqrt{2}}\right) \dyad{\phi} \pm \frac{1}{\sqrt{2}}\sin(\sqrt{2}\tau_s) v_{s,i}.
    \end{equation*}
    By the linearity of the projection superoperator $\mathcal{P}_{T_{C_m}}$, we can evaluate its action on each of these three terms independently. 
    First, applying the projector definition from~\eqref{eq:projector_def} to the base state $C_m$, we obtain $\mathcal{P}_{T_{C_m}}(C_m) = C_m C_m (\mathbb{I}-C_m) + (\mathbb{I}-C_m) C_m C_m = 0$, which vanishes because $C_m(\mathbb{I}-C_m) = 0$. 
    Second, the state $\ket{\phi}$ characterizes the tangent vector $v_{s,i} \in T_{C_m}\mathcal{M}$, meaning it is  orthogonal to the base state, satisfying $C_m \ket{\phi} = 0$. Therefore, evaluating the curvature component yields $\mathcal{P}_{T_{C_m}}(\dyad{\phi}) = C_m \dyad{\phi} (\mathbb{I}-C_m) + (\mathbb{I}-C_m) \dyad{\phi} C_m = 0$. 
    Finally, because the eigenvector $v_{s,i}$ already resides  within the tangent space $T_{C_m}\mathcal{M}$, the orthogonal projector acts upon it as the identity map, yielding $\mathcal{P}_{T_{C_m}}(v_{s,i}) = v_{s,i}$.
    
    Summing these results, the even-parity terms containing the base state and curvature components are annihilated, isolating the odd-parity tangent generator as
    \begin{equation}
        \mathcal{P}_{T_{C_m}}(A_{s,i}^{\pm}) = \pm \frac{1}{\sqrt{2}}\sin(\sqrt{2}\tau_s) v_{s,i}.
    \end{equation}
    Substituting the adaptive geodesic step size $\tau_s = 1/\sqrt{\lambda_{s-1}}$ yields the desired result.
\end{proof}

Consequently, as discussed in the previous section, this is the odd
term~\eqref{eq:odd_term} that we use to estimate the projected component
$\Delta^{(m)}_*$ through the linear reward $Y_s$ in
\eqref{eq:exact_linear}. Therefore, the updates to the tangent design superoperator
are not built from the full measurement directions $A_{s,i}^{\pm}$, but from
their tangent projections
$\mathcal P_{T_{C_m}}(A_{s,i}^{\pm})$.

Using the above lemma, the rank-one update associated with the direction
$v_{s,i}$ is the linear operator on $T_{C_m}\mathcal M$ defined by
\begin{equation}\label{eq:tensor_contribution}
    \mathcal U_{s,i}^{\mathrm{tan}}(X)
    :=
    \left\langle
    \mathcal P_{T_{C_m}}(A_{s,i}^{\pm}),X
    \right\rangle
    \mathcal P_{T_{C_m}}(A_{s,i}^{\pm})
    =
    \frac{1}{2}
    \sin^2\left(\sqrt{\frac{2}{\lambda_{s-1}}}\right)
    \langle v_{s,i},X\rangle v_{s,i},
    \qquad
    X\in T_{C_m}\mathcal M .
\end{equation}
The contribution is independent of the sign $\pm$, since the projected
directions differ only by a sign.

The tangent design superoperator at the end of epoch $m\in[M]$ is therefore
\begin{equation}\label{eq:update_tangent_superoperator}
    \mathcal V^{\mathrm{tan}}_{T_m}(X)
    =
    \mu_{m-1}X
    +
    \sum_{s=1}^{T_m}
    \sum_{i=1}^{d_{\mathrm{tan}}}
    \omega_{s,i}\,
    \mathcal U_{s,i}^{\mathrm{tan}}(X),
    \qquad
    X\in T_{C_m}\mathcal M .
\end{equation}
In the original analysis of linear bandits with vanishing noise
of~\cite{pmlr-v247-lumbreras24a}, controlling the eigenvalues of the design
matrix required a separate matrix-induction argument. In the present
tangent-space construction, this step is simpler because
$\mathcal V_{s-1}^{\mathrm{tan}}$ is a self-adjoint positive operator on the
$d_{\mathrm{tan}}=2(d-1)$-dimensional real Hilbert space
$T_{C_m}\mathcal M$. By the spectral theorem, it admits a
Frobenius-orthonormal eigenbasis
$\{v_{s,i}\}_{i=1}^{d_{\mathrm{tan}}}$ of $T_{C_m}\mathcal M$, such that
\begin{equation}\label{eq:identity_resolution}
    \sum_{i=1}^{d_{\mathrm{tan}}}
    \langle v_{s,i},X\rangle v_{s,i}
    =
    X .
\end{equation}
By fixing a uniform weight within each epoch,
\begin{equation}
    \omega_{s,i}:=\omega_m,
    \qquad s\in[T_m],\quad i\in[d_{\mathrm{tan}}],
\end{equation}
the tangent-space update is isotropic. Indeed, using~\eqref{eq:tensor_contribution}
and the resolution of the identity~\eqref{eq:identity_resolution}, for every
$X\in T_{C_m}\mathcal M$ we have
\begin{equation}
    \sum_{i=1}^{d_{\mathrm{tan}}}
    \omega_m\mathcal U_{s,i}^{\mathrm{tan}}(X)
    =
    \frac{\omega_m}{2}
    \sin^2\left(\sqrt{\frac{2}{\lambda_{s-1}}}\right)
    \sum_{i=1}^{d_{\mathrm{tan}}}
    \langle v_{s,i},X\rangle_F \, v_{s,i}
    =
    \frac{\omega_m}{2}
    \sin^2\left(\sqrt{\frac{2}{\lambda_{s-1}}}\right)
    X .
\end{equation}
Equivalently, the total update at round $s$ is a scalar multiple of the
identity superoperator on $T_{C_m}\mathcal M$
\begin{equation}\label{eq:istropic_contribution}
    \sum_{i=1}^{d_{\mathrm{tan}}}
    \omega_m\mathcal U_{s,i}^{\mathrm{tan}}
    =
    \frac{\omega_m}{2}
    \sin^2\left(\sqrt{\frac{2}{\lambda_{s-1}}}\right)
    \mathcal I_{T_{C_m}\mathcal M}.
\end{equation}
Since the hot-start design satisfies
$\mathcal V_0^{\mathrm{tan}}=\mu_{m-1}\mathcal I_{T_{C_m}\mathcal M}$,
it follows by induction that $\mathcal V_s^{\mathrm{tan}}$ remains a scalar
multiple of the identity throughout the epoch. Thus all its eigenvalues are
equal. We denote their common value by $\lambda_s$, so that
\begin{equation}
    \mathcal V_s^{\mathrm{tan}}
    =
    \lambda_s\mathcal I_{T_{C_m}\mathcal M},
    \qquad
    \lambda_s
    =
    \lambda_{\min}\!\left(\mathcal V_s^{\mathrm{tan}}\right)
    =
    \lambda_{\max}\!\left(\mathcal V_s^{\mathrm{tan}}\right).
\end{equation}
Moreover, the common eigenvalue evolves according to
\begin{equation}
    \lambda_s
    =
    \lambda_{s-1}
    +
    \frac{\omega_m}{2}
    \sin^2\left(\sqrt{\frac{2}{\lambda_{s-1}}}\right),
    \qquad
    \lambda_0=\mu_{m-1}.
\end{equation}
In particular, once the hot-start design is isotropic, any Frobenius-orthonormal
basis of $T_{C_m}\mathcal M$ is an eigenbasis at every step. In Algorithm~\ref{alg:qudit_psmaqb}
we therefore fix one tangent basis of $T_{C_m} \mathcal{M}$ at the beginning of
the epoch and reuse it for all $s$.

This scalar recursion is the only eigenvalue growth calculation needed below.

 Before analyzing the eigenvalue growth, we must specify how the design superoperator is initialized at the start of each epoch. When the algorithm transitions to a new epoch $m$ and updates the base state $C_m$, the tangent space changes. To operate within this new coordinate system, the intrinsic design superoperator must be re-initialized. However, simply resetting it to a fixed regularization $\lambda_0$ would discard the statistical confidence accumulated during the previous epoch. 

To prevent this, we introduce a \emph{hot start}. We transfer only a scalar
precision, not tangent vectors themselves: if
$\mu_{m-1}:=\lambda_{\min}(\mathcal V_{m-1,T_{m-1}}^{\mathrm{tan}})$ is the
final precision from epoch $m-1$, then at the beginning of epoch $m$ we set
\begin{equation}
    \mathcal V_{m,0}^{\mathrm{tan}}
    :=
    \mu_{m-1}\mathcal I_{T_{C_m}\mathcal M}.
\end{equation}
This initialization is an operator on the new tangent space $T_{C_m}\mathcal M$
and does not identify different tangent spaces; it only carries forward the
scalar lower bound on precision.

Now we can establish the eigenvalue growth of the tangent design superoperator by fixing the weights as we will fix them in our algorithm later.

\begin{theorem}\label{thm:isotropic_growth}
Fix an epoch $m\in [M]$. Suppose the tangent design superoperator is hot-started as $
        \mathcal V^{\mathrm{tan}}_0
        =
        \mu_{m-1}\mathcal I_{T_{C_m}\mathcal M}$,
    with inherited precision $\mu_{m-1}\ge 2$. Suppose also that the tangent superoperator $\mathcal{V}^{\mathrm{tan}}_s$ is updated as in~\eqref{eq:update_tangent_superoperator}, and
    the weights are fixed within epoch $m$, namely
    $\omega_{s,i}=\omega_m$ for all $s\in[T_m]$ and
    $i\in[d_{\mathrm{tan}}]$. Then $\mathcal V_s^{\mathrm{tan}}$ remains a
    scalar multiple of the identity on $T_{C_m}\mathcal M$. We write its
    common eigenvalue as $
        \lambda_s
        =
        \lambda_{\min}\!\left(\mathcal V_s^{\mathrm{tan}}\right)
        =
        \lambda_{\max}\!\left(\mathcal V_s^{\mathrm{tan}}\right)$.
    Moreover, let $c_0^2=\sin^2(1)$, and choose
    $\omega_m=\mu_{m-1}/\beta_{\mathrm{var}}$ for some
    $\beta_{\mathrm{var}}\ge 1$. Then
    \begin{equation} \label{eq:isotropic_growth}
        \mu_{m-1}^2 + 2 c_0^2\omega_m s
        \le
        \lambda_s^2
        \le
        \mu_{m-1}^2 + 3\omega_m s .
    \end{equation}
\end{theorem}

\begin{proof}
By the calculation above~\eqref{eq:istropic_contribution}, the tangent design superoperator remains a scalar
multiple of the identity superoperator on $T_{C_m}\mathcal M$ throughout the
epoch. Hence, for every $s\in[T_m]$,
\begin{equation}
    \mathcal V_s^{\mathrm{tan}}
    =
    \lambda_s \mathcal I_{T_{C_m}\mathcal M},
\end{equation}
where the common eigenvalue satisfies the exact recursion 
\begin{equation}
    \lambda_s-\lambda_{s-1}
    =
    \frac{\omega_m}{2}
    \sin^2\!\left(
        \sqrt{\frac{2}{\lambda_{s-1}}}
    \right),
    \qquad
    \lambda_0=\mu_{m-1}.
\end{equation}
It remains to control this scalar recursion.

We start first by proving the lower bound. Fix first the substitution $x := \sqrt{2/\lambda_{s-1}}$. Because $\lambda_{s-1} \ge \mu_{m-1} \ge 2$, we have that $x$ lies in $(0, 1]$. Using the strict concavity bound $\sin(x) \ge \sin(1)x$ for $x \in (0,1]$, we square both sides to obtain $\frac{1}{2}\sin^2(x) \ge \frac{1}{2}\sin^2(1) \left(\frac{2}{\lambda_{s-1}}\right) = \frac{c_0^2}{\lambda_{s-1}}$. Thus, $\lambda_s - \lambda_{s-1} \ge \frac{c_0^2 \omega_m}{\lambda_{s-1}}$. Multiplying by $2\lambda_{s-1}$ yields $2\lambda_{s-1}(\lambda_s - \lambda_{s-1}) \ge 2 c_0^2 \omega_m$. Because $\lambda_s \ge \lambda_{s-1}$, we have $\lambda_s^2 - \lambda_{s-1}^2 = (\lambda_s + \lambda_{s-1})(\lambda_s - \lambda_{s-1}) \ge 2\lambda_{s-1}(\lambda_s - \lambda_{s-1})$. Thus, we obtain the strict lower bound $\lambda_s^2 - \lambda_{s-1}^2 \ge 2 c_0^2 \omega_m$. Summing this telescoping series from $l=1$ to $s$ we have
\begin{align}
 \lambda^2_s - \lambda^2_0  =   \sum_{l=1}^s (\lambda^2_{l} - \lambda^2_{l-1}) \geq  2s c_0^2 \omega_m .
\end{align}
And the lower bound follows using $\lambda_0 = \mu_{m-1}$.
    
Now we prove the upper bound. Using the standard inequality $\sin(x) \le x$ for $x\in [ 0 , \pi /2 ]$, we have $\frac{1}{2}\sin^2(x) \le \frac{1}{\lambda_{s-1}}$. Thus, the increment is bounded by $\lambda_s - \lambda_{s-1} \le \frac{\omega_m}{\lambda_{s-1}}$. We bound the squared difference $\lambda_s^2 - \lambda_{s-1}^2 = (\lambda_s - \lambda_{s-1})(2\lambda_{s-1} + \lambda_s - \lambda_{s-1})$. Substituting the increment bound gives 
\begin{align}
\lambda_s^2 - \lambda_{s-1}^2 \le \frac{\omega_m}{\lambda_{s-1}} \left( 2\lambda_{s-1} + \frac{\omega_m}{\lambda_{s-1}} \right) = 2\omega_m + \frac{\omega_m^2}{\lambda_{s-1}^2}
\end{align}
Then use that $\omega_m = \mu_{m-1} / \beta_{\mathrm{var}}$. To bound the final term, we leverage the fact that the running eigenvalue is monotonically increasing, so $\lambda_{s-1} \ge \mu_{m-1}$, and our initialization design ensures $\mu_{m-1} \ge 1$. This explicitly implies:
\begin{align}
\frac{\omega_m^2}{\lambda_{s-1}^2} \le \frac{(\mu_{m-1}/\beta_{\mathrm{var}})^2}{\mu_{m-1}^2} = \frac{1}{\beta_{\mathrm{var}}^2} \le \frac{1}{\beta_{\mathrm{var}}} \le \frac{\mu_{m-1}}{\beta_{\mathrm{var}}} = \omega_m.
\end{align}
This bounds the step-wise squared growth by $\lambda_s^2 - \lambda_{s-1}^2 \le 3\omega_m$. Summing as before with a telescoping sum to $s$ yields the upper bound.
\end{proof}

\section{Tangent-Space Median of Means (MoM) Estimator}\label{sec:mom_qudit}

In this section, we construct robust confidence regions for the tangent linear estimator least squares defined in~\eqref{eq:exact_linear}. We note that we introduced the adaptive weights $\omega_m$ in order to take advantage of the vanishing statistical noise when measuring in the direction of the unknown state.

Recall that $d_{\mathrm{tan}}=2(d-1)$. For each epoch $m$, step $s\in[T_m]$, tangent direction
$i\in[ d_{\mathrm{tan}}]$, and repetition $j\in[N]$, let
$X^{\pm}_{m,s,i,j}\in\{0,1\}$ denote the two binary outcomes obtained by measuring the
symmetric rank-one projectors $A^{\pm}_{m,s,i}$ on the unknown state $\rho\in\mathcal{S}^*_d$. We define
\begin{align}\label{eq:reward_tangent}
    Y^{(j)}_{m,s,i}
    :=
    \frac{1}{2}
    \left(
        X^+_{m,s,i,j}
        -
        X^-_{m,s,i,j}
    \right).
\end{align}
Let $\mathcal H_m$ be the $\sigma$-algebra generated by the data, any
algorithmic randomization used before epoch $m$, and all actions and outcomes from
epochs $1,\ldots,m-1$. In particular, the base state $C_m$, the inherited precision
$\mu_{m-1}$, and the epoch weight $\omega_m$ are $\mathcal H_m$-measurable. Let $\mathbf H_m$ denote the complete random history before epoch $m$.
Thus $\mathcal H_m=\sigma(\mathbf H_m)$ is the $\sigma$-algebra generated
by the data available at the start of epoch $m$.

For a step $s$ inside epoch $m$, write
\begin{align}
    \mathbf X_{m,s}
    :=
    \left(
        X^+_{m,s,i,j},
        X^-_{m,s,i,j}
    \right)_{i\in[d_{\mathrm{tan}}],\,j\in[N]}
\end{align}
for the collection of fresh binary outcomes observed at step $s$. For $s\in[T_m]$, we define the pre-measurement $\sigma$-algebra at step
$s$ by
\begin{align}\label{eq:pre_measurement_filtration}
    \mathcal F_{m,s-1}
    :=
    \sigma\!\left(
        \mathbf H_m,
        \mathbf X_{m,1},
        \ldots,
        \mathbf X_{m,s-1}
    \right).
\end{align}
For $s=1$, this means $\mathcal F_{m,0}=\mathcal H_m$.

Thus $\mathcal F_{m,s-1}$ represents everything that is known immediately before the measurements
at step $s$ of epoch $m$. In particular, the actions
$A^\pm_{m,s,i}$ and tangent generators $O_{m,s,i}$ have already been
chosen and are $\mathcal F_{m,s-1}$-measurable, while the fresh outcomes
$X^\pm_{m,s,i,j}$ from step $s$ are not included in
$\mathcal F_{m,s-1}$.

 Conditionally on $\mathcal F_{m,s-1}$, the actions are fixed and the fresh
outcomes satisfy
\begin{align}
    \mathbb E\!\left[
        X^\pm_{m,s,i,j}
        \mid
        \mathcal F_{m,s-1}
    \right]
    =
    \operatorname{Tr}\!\left(\rho A^\pm_{m,s,i}\right).
\end{align}

We define
\begin{align}\label{eq:mom_noise_def}
    \eta^{\pm}_{m,s,i,j}
    :=
    X^{\pm}_{m,s,i,j}
    -
    \operatorname{Tr}
    \left(
        \rho A^{\pm}_{m,s,i}
    \right),
    \qquad
    \varepsilon^{(j)}_{m,s,i}
    :=
    \frac{1}{2}
    \left(
        \eta^+_{m,s,i,j}
        -
        \eta^-_{m,s,i,j}
    \right).
\end{align}
Using the definition~\eqref{eq:reward_tangent}
and the identity from~\eqref{eq:exact_linear}, we obtain
\begin{align}\label{eq:mom_linear_model}
    Y^{(j)}_{m,s,i}
    =
    \left\langle
        \Delta_*^{(m)},O_{m,s,i}
    \right\rangle
    +
    \varepsilon^{(j)}_{m,s,i}.
\end{align}
Since the actions are fixed conditionally on $\mathcal F_{m,s-1}$ and the
binary outcomes have conditional means
$\operatorname{Tr}(\rho A^{\pm}_{m,s,i})$, we have
\begin{align}\label{eq:epsilon_filtration}
    \mathbb E
    \left[
        \varepsilon^{(j)}_{m,s,i}
        \middle|
        \mathcal F_{m,s-1}
    \right]
    =
    0.
\end{align}
And, by conditional independence of the two fresh-copy measurements,
\begin{align}
\operatorname{Var}
\left(
    \varepsilon^{(j)}_{m,s,i}
    \middle|
    \mathcal F_{m,s-1}
\right)
=
\frac{1}{4}
\left[
\operatorname{Var}
\left(
    X^+_{m,s,i,j}
    \middle|
    \mathcal F_{m,s-1}
\right)
+
\operatorname{Var}
\left(
    X^-_{m,s,i,j}
    \middle|
    \mathcal F_{m,s-1}
\right)
\right].
\end{align}
All conditional expectations and variances in the qudit epoch analysis are understood
with respect to this filtration. When the repetition index $j$ is irrelevant, we
suppress it and write $\varepsilon_{m,s,i}$.

The key variance-normalization condition for the MoM argument is
\begin{align}
    \omega_m
    \operatorname{Var}
    \left(
        \varepsilon^{(j)}_{m,s,i}
        \middle|
        \mathcal F_{m,s-1}
    \right)
    \le 1
\end{align}
for every epoch $m$, step $s$, tangent direction $i$, and repetition
$j$. This is the tangent-space analogue of the weighted MoM least-squares
normalization used in the qubit analysis~\cite{lumbreras2024learning}. Once
this condition holds, the random part of each least-squares estimator has a
controlled second moment, which is enough to combine Markov's inequality with
the distance-median MoM selector. The proof below follows the general MoM
strategy of~\cite{heavy_tail_linear_optimal}, with modifications needed
because the estimators live in the tangent space $T_{C_m}\mathcal M$ and the
design operator is a tangent-space superoperator.

In this subsection we analyze a single epoch $m$. All quantities below are understood
to belong to this fixed epoch, and we suppress the epoch subscript $m$ to lighten
notation. Thus we write
\begin{align}
    C=C_m,\qquad
    \mu=\mu_{m-1},\qquad
    \omega=\omega_m,\qquad
    T=T_m,
\end{align}
and similarly
\begin{align}
    A^\pm_{s,i}=A^\pm_{m,s,i},\qquad
    O_{s,i}=O_{m,s,i},\qquad
    Y^{(j)}_{s,i}=Y^{(j)}_{m,s,i},\qquad
    \varepsilon^{(j)}_{s,i}=\varepsilon^{(j)}_{m,s,i}.
\end{align}
We condition on the history before epoch $m$, so that $C_m$, $\mu_{m-1}$,
$\omega_m$, and the rule used to choose the within-epoch actions are fixed. The
filtration $\mathcal F_{s-1}$ (which should be read as
$\mathcal F_{m,s-1}$) denotes the information available just before the
measurements at step $s$ of this fixed epoch: it contains the pre-epoch history, all
actions and outcomes from steps $1,\ldots,s-1$, and the tangent directions and
actions chosen for step $s$, but not the fresh outcomes at step $s$.

Following the distance-median MoM construction used for the qubit estimator
in~\cite{lumbreras2024learning}, we build parallel estimators in each epoch.
Specifically, we repeat each chosen tangent action
$N = 2 \lceil 12 \log(T_{\mathrm{total}}/\delta) \rceil$ times at every step
$s \in [T_m]$. By the symmetric-pair construction above, this yields
$N$ independent difference-rewards
$Y_{s,i}^{(1)}, \dots, Y_{s,i}^{(N)}$ for the same tangent generator
$O_{s,i}$.

Using these independent samples, we construct $N$ parallel least-squares
estimators within the real tangent space $T_{C_m}\mathcal M$, whose dimension
is $d_{\mathrm{tan}}=2(d-1)$. Since the action sequence and the weights are
the same across the $N$ parallel streams, the estimators share the same
tangent design superoperator. This operator is hot-started with the inherited
precision $\mu_{m-1}$ from the previous epoch.
More precisely, let $O_{s,i}\in T_{C_m}\mathcal M$ denote the tangent
measurement direction used at round $s$ and direction $i$. We define
$\mathcal V_{T_m}^{\mathrm{tan}}:T_{C_m}\mathcal M\to T_{C_m}\mathcal M$ by
\begin{equation}\label{eq:final_mom_design}
    \mathcal V_{T_m}^{\mathrm{tan}}(X)
    :=
    \mu_{m-1}X
    +
    \sum_{s=1}^{T_m}
    \sum_{i=1}^{d_{\mathrm{tan}}}
    \omega_m
    \langle O_{s,i},X\rangle O_{s,i},
    \qquad
    X\in T_{C_m}\mathcal M .
\end{equation}

The $j$-th independent estimator for
$\Delta_*^{(m)}=\mathcal P_{T_{C_m}}(\rho-C_m)$ is then
\begin{equation}\label{eq:mom_local_est}
    \widehat{\Delta}^{(j)}
    :=
    \left(\mathcal V_{T_m}^{\mathrm{tan}}\right)^{-1} \left(
    \sum_{s=1}^{T_m}
    \sum_{i=1}^{d_{\mathrm{tan}}}
    \omega_m Y_{s,i}^{(j)} O_{s,i} \right),
    \qquad
    j\in[N].
\end{equation}

We aggregate these estimators using a distance-median rule in the design norm,
following the weighted MoM principle of~\cite{heavy_tail_linear_optimal}.
Because our local estimators are matrices in $T_{C_m}\mathcal{M}$, the weighted
norm is defined through the Frobenius inner product: for any tangent matrix
$X$,
\begin{align}
\|X\|_{\mathcal{V}_{T_m}^{\text{tan}}} = \sqrt{\langle X, \mathcal{V}_{T_m}^{\text{tan}}(X) \rangle}.
\end{align}
We compute then distance from each estimator to all other estimators using this weighted norm, find the median distance for each, and select the estimator that minimizes this median distance. In short, we compute the estimator $\hat{\Delta}_{\text{MoM}}^{(m)}$ using
\begin{equation} \label{eq:mom_selector}
\hat{\Delta}_{\text{MoM}}^{(m)} := \hat{\Delta}^{(j^*)}, \quad j^* = \argmin_{j \in [N]} y_j , \quad    y_j = \operatorname{median}_{l \neq j} \big\| \hat{\Delta}^{(j)} - \hat{\Delta}^{(l)} \big\|_{\mathcal{V}_{T_m}^{\text{tan}}}  .
\end{equation}
The following theorem formalizes the concentration bound of this MoM estimator. The original arguments of the proof can be found in~\cite[Lemma 2 and 3]{heavy_tail_linear_optimal} while here we had to make some adaptations and use properties of our explicit procedure.

\begin{theorem}
\label{thm:mom_bound}
Fix an epoch $m$ and condition on the start-of-epoch history
$\mathcal H_m$. During epoch $m$, the base state $C_m$, the tangent
space $T_{C_m}\mathcal M$, the inherited precision $\mu_{m-1}$, the
epoch weight $\omega_m$, and the epoch length $T_m$ are fixed. Let $
    \Delta_*^{(m)}
    :=
    \mathcal P_{T_{C_m}}(\rho-C_m)
    \in T_{C_m}\mathcal M$
be the local tangent target, as defined in \eqref{eq:tangent_parameter}. Let $N\in\mathbb{N}$ be the number of repetitions we use in the MoM procedure. For
each repetition block $j\in[N]$, let $\widehat{\Delta}^{(j)}$ be the
tangent least-squares estimator defined in \eqref{eq:mom_local_est}, and
let $\hat{\Delta}_{\mathrm{MoM}}^{(m)}$ be the distance-median estimator
defined in \eqref{eq:mom_selector}. Let
$\mathcal V^{\mathrm{tan}}_{T_m}$ denote the final tangent design
superoperator of epoch $m$~\eqref{eq:final_mom_design}. 

Assume that the variance-normalization condition
\begin{align}
    \omega_m
    \operatorname{Var}
    \left(
        \varepsilon^{(j)}_{m,s,i}
        \,\middle|\,
        \mathcal F_{m,s-1}
    \right)
    \le 1
\end{align}
holds almost surely for every $s\in[T_m]$, $i\in[d_{\mathrm{tan}}]$, and
$j\in[N]$, where the noise variables $\varepsilon^{(j)}_{m,s,i}$ are
defined in \eqref{eq:mom_noise_def}. Define the regularization bias and the biased target
\begin{align}
    B_m
    :=
    \mu_{m-1}
    \left(
        \mathcal V^{\mathrm{tan}}_{m,T_m}
    \right)^{-1}
    \Big(\Delta_*^{(m)}\Big) , \quad
    \Delta'_*
    :=
    \Delta_*^{(m)}-B_m .
\end{align}
Then, conditioned on $\mathcal H_m$,
\begin{align}
    \Pr\left(
        \left\|
            \hat{\Delta}_{\mathrm{MoM}}^{(m)}
            -
            \Delta'_*
        \right\|^2_{\mathcal V^{\mathrm{tan}}_{T_m}}
        \le
        72(d-1)
        \,\middle|\,
        \mathcal H_m
    \right)
    \ge
    1-\exp(-N/8).
\end{align}
In particular, for $
    N=2\left\lceil 12\log\left(T_{\mathrm{total}}/\delta\right)\right\rceil$
the conditional failure probability is at most
$\left(\delta/T_{\mathrm{total}}\right)^3$.
\end{theorem}

\begin{proof}
    Recall that we fixed the base state $C_m$, and the combined outcomes from the measurements~\eqref{eq:mom_linear_model} which are $Y_{s,i}^{(j)} = \langle \Delta_*^{(m)}, O_{s,i} \rangle + \epsilon_{s,i}^{(j)}$. Substituting them into the $j$-th estimator~\eqref{eq:mom_local_est} gives
\begin{equation}\label{eq:error_tangent_par}
    \widehat{\Delta}^{(j)}
    =
    \Delta_*^{(m)}
    -
    B_m
    +
    Z^{(j)},
\end{equation}
where
\begin{equation}
    Z^{(j)}
    :=
    \left(\mathcal V_{T_m}^{\mathrm{tan}}\right)^{-1} \left(
    \sum_{s=1}^{T_m}
    \sum_{i=1}^{d_{\mathrm{tan}}}
    \omega_m
    \varepsilon_{m,s,i}^{(j)}
    O_{m,s,i} \right),
    \qquad
    B_m
    :=
    \mu_{m-1}
    \left(\mathcal V_{T_m}^{\mathrm{tan}}\right)^{-1}
    (\Delta_*^{(m)} ).
\end{equation}
  Since $Z^{(j)}$ contains the random part of the estimator through the noise variables
$\varepsilon_{m,s,i}^{(j)}$, we first bound its conditional second moment, $
    \mathbb E\!\left[
        \|Z^{(j)}\|_{\mathcal V_{T_m}^{\mathrm{tan}}}^2
        \,\middle|\,
        \mathcal H_m
    \right]$.
This will then be converted into a probabilistic bound by Markov's inequality.

    Recall from~\eqref{eq:epsilon_filtration} the measurements are unbiased given the pre-measurement filtration, we have $\mathbb E
    [
        \varepsilon^{(j)}_{m,s,i}
        \mid
        \mathcal F_{m,s-1}
    ]
    =
    0.$
Hence, conditionally on the history $\mathcal H_m$, the term $Z^{(j)}$ has mean
zero i.e 
\begin{align}
    \EX \big[ Z^{(j)} \big| \mathcal{H}_m \big] = 0 .
\end{align}
Let us define the following shorthand notation 
\begin{align}\label{eq:Sj}
     S^{(j)}
    :=
    \sum_{s=1}^{T_m}
    \sum_{i=1}^{d_{\mathrm{tan}}}
    \omega_m
    \varepsilon^{(j)}_{m,s,i}
    O_{m,s,i},
\end{align}
such that $Z^{(j)} = ({\mathcal V_{T_m}^{\mathrm{tan}}})^{-1} ( S^{(j)} )$.
We start by  fixing a tangent vector $G\in T_{C_m}\mathcal M$ and computing the second moment $\mathbb E\!\left[
    \langle G,S^{(j)}\rangle^2
    \,\middle|\,
    \mathcal H_m
\right]$.
We use again~\eqref{eq:epsilon_filtration} that is $\
    \mathbb E[
        \varepsilon^{(j)}_{m,s,i}
        \mid
        \mathcal F_{m,s-1}
    ]
    =
    0$,
together with the fact that $O_{m,s,i}$ is $\mathcal F_{m,s-1}$-measurable to get
\begin{align}
    \mathbb E\!\left[
        \omega_m
        \varepsilon^{(j)}_{m,s,i}
        O_{m,s,i}
        \middle|
        \mathcal F_{m,s-1}
    \right]
    =
    0.
\end{align}
Using~\eqref{eq:Sj} we have
\begin{align}\label{eq:inner_GS}
    \langle G,S^{(j)}\rangle
    =
    \sum_{s=1}^{T_m}
    \sum_{i=1}^{d_{\mathrm{tan}}}
    \omega_m
    \varepsilon^{(j)}_{m,s,i}
    \langle G,O_{m,s,i}\rangle .
\end{align}
We need to compute the square of the above quantity and take expectation. For that first, we compute per separate all the terms and check that all the cross term vanish.

Consider two distinct terms. If $s<s'$, then, using the tower property and~\eqref{eq:epsilon_filtration} we have
\begin{align}
\begin{aligned}
&\mathbb E\!\left[
    \varepsilon^{(j)}_{m,s,i}
    \varepsilon^{(j)}_{m,s',i'}
    \langle G,O_{m,s,i}\rangle
    \langle G,O_{m,s',i'}\rangle
    \,\middle|\,
    \mathcal H_m
\right] \\
&\qquad =
\mathbb E\!\left[
    \varepsilon^{(j)}_{m,s,i}
    \langle G,O_{m,s,i}\rangle
    \langle G,O_{m,s',i'}\rangle
    \mathbb E\!\left[
        \varepsilon^{(j)}_{m,s',i'}
        \middle|
        \mathcal F_{m,s'-1}
    \right]
    \,\middle|\,
    \mathcal H_m
\right]
=
0 .
\end{aligned}
\end{align}
The case $s'>s$ is symmetric. If $s=s'$ and $i\neq i'$, then the
fresh-copy measurements in different tangent directions are conditionally
independent given $\mathcal F_{m,s-1}$, and both have conditional mean zero~\eqref{eq:epsilon_filtration}.
Therefore
\begin{equation}
    \mathbb E\!\left[
        \varepsilon^{(j)}_{m,s,i}
        \varepsilon^{(j)}_{m,s,i'}
        \middle|
        \mathcal F_{m,s-1}
    \right]
    =
    0 .
\end{equation}
Thus only the diagonal terms remain, and hence we square~\eqref{eq:inner_GS} and take expectation to get
\begin{align}\label{eq:gs_expectation}
\begin{aligned}
\mathbb E\!\left[
    \langle G,S^{(j)}\rangle^2
    \,\middle|\,
    \mathcal H_m
\right]
&=
\mathbb E\!\left[
    \sum_{s=1}^{T_m}
    \sum_{i=1}^{d_{\mathrm{tan}}}
    \omega_m^2
    \left(\varepsilon^{(j)}_{m,s,i}\right)^2
    \langle G,O_{m,s,i}\rangle^2
    \,\middle|\,
    \mathcal H_m
\right] \\
&=
\mathbb E\!\left[
    \sum_{s=1}^{T_m}
    \sum_{i=1}^{d_{\mathrm{tan}}}
    \omega_m^2
    \operatorname{Var}\!\left(
        \varepsilon^{(j)}_{m,s,i}
        \middle|
        \mathcal F_{m,s-1}
    \right)
    \langle G,O_{m,s,i}\rangle^2
    \,\middle|\,
    \mathcal H_m
\right]. 
\end{aligned}
\end{align}
Recall the assumed variance normalization, $
    \omega_m
    \operatorname{Var}\!\left(
        \varepsilon^{(j)}_{m,s,i}
        \middle|
        \mathcal F_{m,s-1}
    \right)
    \le
    1$. For every $G\in T_{C_m}\mathcal M$, using~\eqref{eq:gs_expectation} and recalling the definition of the design superoperator in~\eqref{eq:final_mom_design}, we find that 
\begin{align}\label{eq:final_gsbound}
\begin{aligned}
\mathbb E\!\left[
    \langle G,S^{(j)}\rangle^2
    \,\middle|\,
    \mathcal H_m
\right]
&\le
\sum_{s=1}^{T_m}
\sum_{i=1}^{d_{\mathrm{tan}}}
\omega_m
\langle G,O_{m,s,i}\rangle^2                                      \\
&=
\left\langle
    G,
    \left(
        \mathcal V_{T_m}^{\mathrm{tan}}-\mu_{m-1}\mathcal I_{T_{C_m}\mathcal M}
    \right)G
\right\rangle                                                       \\
&\le
\left\langle
    G,\mathcal V_{T_m}^{\mathrm{tan}} G
\right\rangle .
\end{aligned}
\end{align}
We can now compute the second moment of $Z^{(j)}$. We obtain, for every $H\in T_{C_m}\mathcal M$,
\begin{align}\label{eq:inverse_inequality}
\begin{aligned}
\mathbb E\!\left[
    \langle H,Z^{(j)}\rangle^2
    \,\middle|\,
    \mathcal H_m
\right]
&=
\mathbb E\!\left[
    \left\langle
        H,\mathcal (\mathcal{V}_{T_m}^{\mathrm{tan}})^{-1}S^{(j)}
    \right\rangle^2
    \,\middle|\,
    \mathcal H_m
\right]                                                   \\
&=
\mathbb E\!\left[
    \left\langle
        \mathcal (\mathcal{V}_{T_m}^{\mathrm{tan}})^{-1}H,S^{(j)}
    \right\rangle^2
    \,\middle|\,
    \mathcal H_m
\right]                                                   \\
&\le
\left\langle
    \mathcal (\mathcal{V}_{T_m}^{\mathrm{tan}})^{-1}H,
    \mathcal \mathcal{V}_{T_m}^{\mathrm{tan}}\mathcal (\mathcal{V}_{T_m}^{\mathrm{tan}})^{-1}H
\right\rangle                                             \\
&=
\left\langle
    H,\mathcal ( \mathcal{V}_{T_m}^{\mathrm{tan}})^{-1}H
\right\rangle ,
\end{aligned}
\end{align}
where we used the bound
$\mathbb E[\langle G,S^{(j)}\rangle^2\mid\mathcal H_m]\le
\langle G,\mathcal{V}_{T_m}^{\mathrm{tan}} G\rangle$ from~\eqref{eq:final_gsbound} with $G=\mathcal (\mathcal{V}_{T_m}^{\mathrm{tan}})^{-1}H $.

Now we can turn the above into a bound on the full  norm $\| \cdot \|_{\mathcal{V}_{T_m}^{\mathrm{tan}}}$ by expanding $Z^{(j)}$ in
an eigenbasis of $\mathcal \mathcal{V}_{T_m}^{\mathrm{tan}}$.
Let $\{e_r\}_{r=1}^{d_{\mathrm{tan}}}$ be a Frobenius-orthonormal eigenbasis of the tangent-space operator $\mathcal{V}_{T_m}^{\mathrm{tan}}$, meaning that $\mathcal{V}_{T_m}^{\mathrm{tan}} e_r=\lambda_r e_r$ with $\lambda_r>0$ for all $r$. Then
\begin{align}
\begin{aligned}
\mathbb E\!\left[
    \|Z^{(j)}\|_{\mathcal{V}_{T_m}^{\mathrm{tan}}}^2
    \,\middle|\,
    \mathcal H_m
\right]
&=
\mathbb E\!\left[
    \left\langle Z^{(j)},\mathcal{V}_{T_m}^{\mathrm{tan}} Z^{(j)}\right\rangle
    \,\middle|\,
    \mathcal H_m
\right]                                                        \\
&=
\sum_{r=1}^{d_{\mathrm{tan}}}
\lambda_r
\mathbb E\!\left[
    \langle e_r,Z^{(j)}\rangle^2
    \,\middle|\,
    \mathcal H_m
\right]                                                        \\
&\le
\sum_{r=1}^{d_{\mathrm{tan}}}
\lambda_r
\left\langle e_r,\mathcal (\mathcal{V}_{T_m}^{\mathrm{tan}})^{-1}e_r\right\rangle                \\
&=
\sum_{r=1}^{d_{\mathrm{tan}}}
\lambda_r\frac{1}{\lambda_r}
=
d_{\mathrm{tan}}
=
2(d-1),
\end{aligned}
\end{align}
where we have applied~\eqref{eq:inverse_inequality} in the inequality with $H = e_r\in T_{C_m}$ and $Z^{(j)} = \sum_{r=1}^{d_{\mathrm{tan}}}\langle e_r , Z^{(j)} \rangle e_r$. Recall also that $(\mathcal{V}_{T_m}^{\mathrm{tan}})^{-1}$ denotes the inverse of $\mathcal{V}_{T_m}^{\mathrm{tan}}$ as a positive definite self-adjoint operator on $T_{C_m}\mathcal M$ with respect to the Frobenius inner product; hence $\mathcal (\mathcal{V}_{T_m}^{\mathrm{tan}})^{-1}e_r=\lambda_r^{-1}e_r$ and $\langle e_r,\mathcal (\mathcal{V}_{T_m}^{\mathrm{tan}})^{-1}e_r\rangle=1/\lambda_r$.

Having bounded the second moment of $Z^{(j)}$ we are ready to translate it into a probabilistic bound. Applying Markov's inequality conditionally on $\mathcal H_m$, we get
\begin{align}
\begin{aligned}
\Pr\left(
    \|Z^{(j)}\|_{\mathcal{V}_{T_m}^{\mathrm{tan}}}^2
    \ge
    8(d-1)
    \,\middle|\,
    \mathcal H_m
\right)
&\le
\frac{
    \mathbb E\left[
        \|Z^{(j)}\|_{\mathcal{V}_{T_m}^{\mathrm{tan}}}^2
        \mid
        \mathcal H_m
    \right]
}{
    8(d-1)
}                                                    \\
&\le
\frac{2(d-1)}{8(d-1)}
=
\frac14.
\end{aligned}
\end{align}
Equivalently, with conditional probability at least $3/4$ we have
\begin{align}\label{eq:34prob}
  \mathrm{Pr}\left(  \|Z^{(j)}\|^2_{\mathcal{V}_{T_m}^{\mathrm{tan}}}
    \le
    8(d-1) \,\middle|\,
        \mathcal H_m \right) \geq \frac{3}{4}.
\end{align}
    Now that we have bounded the separate estimators we will use the median of means technique to get a better concentration bound.

   Define the good-index set
\begin{align}
    \mathcal G
    :=
    \left\{
        j\in[N]:
        \left\|
            \widehat{\Delta}^{(j)}-\Delta'_*
        \right\|_{\mathcal{V}_{T_m}^{\mathrm{tan}}}
        \le
      \sqrt{8(d-1)}
    \right\}.
\end{align}
Since
\begin{align}
    \widehat{\Delta}^{(j)}-\Delta'_*
    =
    Z^{(j)},
\end{align}
by~\eqref{eq:34prob} each index is conditionally good with probability at least $3/4$. Moreover,
the $N$ parallel streams use independent fresh-copy measurements. Hence,
conditionally on $\mathcal H_m$, the indicators
\begin{align}
    I_j := \mathds 1\{j\in\mathcal G\}
\end{align}
are independent and satisfy
\begin{align}
    \mathbb E[I_j\mid\mathcal H_m]\ge \frac34.
\end{align}
Hoeffding's inequality for independent bounded variables gives
\begin{align}
\begin{aligned}
\Pr\left(
    |\mathcal G|
    \le
    \frac N2
    \,\middle|\,
    \mathcal H_m
\right)
&=
\Pr\left(
    \sum_{j=1}^N I_j
    \le
    \frac N2
    \,\middle|\,
    \mathcal H_m
\right)                                      \\
&\le
\exp(-N/8).
\end{aligned}
\end{align}
Therefore, conditionally on $\mathcal H_m$, with probability at least
$1-\exp(-N/8)$, more than half of the estimators are good. Equivalently,
\begin{align}\label{eq:good_event_mom}
\Pr\left(
    |\mathcal G|>\frac N2
    \,\middle|\,
    \mathcal H_m
\right)
\ge
1-\exp(-N/8),
\end{align}
where
\begin{align}
\mathcal G
=
\left\{
j\in[N]:
\left\|
\widehat{\Delta}^{(j)}-\Delta'_*
\right\|_{\mathcal{V}^{\mathrm{tan}}_{T_m}}
\le
\sqrt{8(d-1)}
\right\}.
\end{align}
    
    Conditioned on this high-probability event, we analyze the distance-median selector defined in~\eqref{eq:mom_selector}. To bound the distance between the selected estimator and the true biased target, we proceed through three logical deductions relying on the triangle inequality.

First, we bound the median distance for any estimator belonging to the good set. For any two good estimators $g_1, g_2 \in \mathcal{G}$, the triangle inequality ensures their distance is bounded by the sum of their individual errors, which gives
\begin{align}
\| \hat{\Delta}^{(g_1)} - \hat{\Delta}^{(g_2)}\|_{\mathcal{V}_{T_m}^{\text{tan}}} \leq \|\hat{\Delta}^{(g_1)} - \Delta'_* \|_{\mathcal{V}_{T_m}^{\text{tan}}}+\| \hat{\Delta}^{(g_2)} - \Delta'_* \|_{\mathcal{V}_{T_m}^{\text{tan}}}\leq 2 \sqrt{8(d-1)}. 
\end{align}
Because by~\eqref{eq:good_event_mom} more than half of all the estimators are good, the median distance from any one good estimator $g$ to the rest must be bounded by this maximum distance to another good estimator. Recall that this median distance for estimator $g$ is defined as $\text{median}_{l \neq g} \| \hat{\Delta}^{(g)} - \hat{\Delta}^{(l)} \|$. Therefore, we establish that
\begin{align}
\text{median}_{l \neq g} \| \hat{\Delta}^{(g)} - \hat{\Delta}^{(l)} \|_{\mathcal{V}_{T_m}^{\text{tan}}} \leq 2 \sqrt{8(d-1)}.
\end{align}

Next, we consider the specific estimator chosen by our algorithm, denoted as $j^*$. By definition, the algorithm selects the index $j^*$ that minimizes this median distance across all the estimators. Since the median distance for our selected estimator must be less than or equal to the median distance of any good estimator, we substitute the explicit norms directly to find
\begin{align}
\text{median}_{l \neq j^*} \| \hat{\Delta}^{(j^*)} - \hat{\Delta}^{(l)} \|_{\mathcal{V}_{T_m}^{\text{tan}}} \leq \text{median}_{l \neq g} \| \hat{\Delta}^{(g)} - \hat{\Delta}^{(l)} \|_{\mathcal{V}_{T_m}^{\text{tan}}} \leq 2 \sqrt{8(d-1)}.
\end{align}

We now bridge the gap between our selected estimator and the target parameter $\Delta'_*$. Let $\mathcal{S}_{j^*}$ denote the subset of estimators that fall within the median distance from $j^*$ which is
\begin{align}
\mathcal{S}_{j^{*}} := \left\{ l \in [N] : \| \hat{\Delta}^{(j^*)} - \hat{\Delta}^{(l)} \|_{\mathcal{V}_{T_m}^{\text{tan}}} \leq \text{median}_{k \neq j^*} \| \hat{\Delta}^{(j^*)} - \hat{\Delta}^{(k)} \|_{\mathcal{V}_{T_m}^{\text{tan}}} \right\}.
\end{align}
By definition of the median over the \(N-1\) distances with \(\ell\neq j^*\),
the set \(S_{j^*}\) contains at least \(N/2\) indices. On the event
\(|\mathcal G|>N/2\), we therefore have
\begin{align}
    |S_{j^*}|+|\mathcal G| > N .
\end{align}
Hence $S_{j^*}\cap \mathcal G\neq\emptyset$.

This guarantees the existence of at least one estimator $g^*$ such that $g^* \in \mathcal{S}_{j^*}$ and $g^* \in \mathcal{G}$. By the definition of the set $\mathcal{S}_{j^*}$, the fact that $g^*$ is an element of this set directly implies that its distance from $j^*$ satisfies
\begin{align}
\| \hat{\Delta}^{(j^*)} - \hat{\Delta}^{(g^*)} \|_{\mathcal{V}_{T_m}^{\text{tan}}} \leq 2 \sqrt{8(d-1)}.
\end{align}

Recall from~\eqref{eq:good_event_mom} that we can bound the distance in the good respect to the biased target $\Delta'_* :=\Delta_*^{(m)} - B_m$. Then applying triangle inequalities against  $j^*$ we get
\begin{align}
\| \hat{\Delta}^{(j^*)} - \Delta'_* \|_{\mathcal{V}_{T_m}^{\text{tan}}} &\leq \| \hat{\Delta}^{(j^*)} - \hat{\Delta}^{(g^*)} \|_{\mathcal{V}_{T_m}^{\text{tan}}} + \| \hat{\Delta}^{(g^*)} - \Delta'_* \|_{\mathcal{V}_{T_m}^{\text{tan}}} \nonumber \\
&\leq 2 \sqrt{8(d-1)} + \sqrt{8(d-1)} \nonumber \\
&= 3 \sqrt{8(d-1)}.
\end{align}

Squaring this final radius yields the purely statistical confidence bound $(3 \sqrt{8(d-1)})^2 = 72(d-1)$, which completes the proof.

\end{proof}

\section{Qudit algorithm}\label{sec:qudit_algorithm}

We now have all the required tools to state the epoch-based algorithm for qudit pure-state tomography with minimal regret. The algorithm freezes the base state $C_m$ during each epoch, performs all linear estimation in the tangent space $T_{C_m}\mathcal{M}$, and carries only the scalar precision $\mu_m$ to the next epoch. This prevents mixing tangent vectors defined at different base points. 

We use two indices for the tangent design superoperator. The index $m$ labels the epoch, while~$s$ labels the step inside that epoch. Thus,
\begin{equation}
    \mathcal{V}^{\mathrm{tan}}_{m,s}
    :
    T_{C_m}\mathcal{M}
    \to
    T_{C_m}\mathcal{M}
\end{equation}
denotes the tangent design superoperator after $s$ updates in epoch $m$. At the start of the epoch we hot-start it as
\begin{equation}
    \mathcal{V}^{\mathrm{tan}}_{m,0}
    =
    \mu_{m-1}
    \mathcal{I}_{T_{C_m}\mathcal{M}}.
\end{equation}
Then, for $s\geq 1$, the superoperator $\mathcal{V}^{\mathrm{tan}}_{m,s-1}$ is the design superoperator available before choosing the actions at step $s$. We denote its smallest eigenvalue by
\begin{equation}
    \lambda_{m,s-1}
    :=
    \lambda_{\min}
    \left(
    \mathcal{V}^{\mathrm{tan}}_{m,s-1}
    \right).
\end{equation}
At the end of the epoch, the precision passed to the next epoch is
\begin{equation}
    \mu_m
    :=
    \lambda_{\min}
    \left(
    \mathcal{V}^{\mathrm{tan}}_{m,T_m}
    \right).
\end{equation}
This notation emphasizes that the design superoperator is always defined in the tangent space of the current base state $C_m$, and that only the scalar precision $\mu_m$ is transferred to the next epoch.

For the warm-up phase, we only need a constant-accuracy estimate of the
unknown pure state. We use any standard tomography subroutine based on
rank-one two-outcome measurements which, with probability at least
$1-\delta_w$, returns a pure state $C_1\in\mathcal S_d^*$ satisfying
\begin{align}\label{eq:warmup_condition_qudit}
    \|\rho-C_1\|_F^2 \le \frac14 .
\end{align}
For example, this can be obtained by linear inversion or projected
least-squares tomography with rank-one measurement directions drawn from
a complex projective design, followed by projection onto
$\mathcal S_d^*$; see, e.g., the projected least-squares tomography
guarantees for 2-design/uniform rank-one measurements in
\cite{guctua2020fast} and the rank-one design recovery results of
\cite{kueng2017low}. Since this step is used only to reach a fixed
constant accuracy, we denote its sample cost by
\begin{equation}\label{eq:warmup_sample_cost_qudit}
    T_0
    =
    \mathcal{O}\!\left(
    d^2\log(1/\delta_{\mathrm{w}})
    \right).
\end{equation}
We denote this event by
\begin{equation}\label{eq:event_G0_qudit}
    \mathcal{G}_0
    :=
    \left\{
    \|\rho-C_1\|_F^2
    \leq
    \frac{1}{4}
    \right\}.
\end{equation}

We start by defining all the constants used during the algorithm and that will be used through the formal proof of the regret analysis and its auxiliary lemmas.
Set
\begin{align}\label{eq:qudit_alg_parameters}
d_{\mathrm{tan}}&:=2(d-1),\qquad
N:= 2\left \lceil 12\log(T_{\mathrm{total}}/\delta)\right\rceil,\qquad
c_0^2:=\sin^2(1), \nonumber\\
\beta_{\mathrm{stat}}&:=72(d-1),\qquad
\beta_{\max}:=4\beta_{\mathrm{stat}},\qquad
L_r:= 6, \nonumber\\
\beta_{\mathrm{var}}&:=L_r^2\beta_{\max}+1,
\qquad
\alpha:=\left\lceil \frac{8L_r^4\beta_{\mathrm{var}}}{c_0^2}\right\rceil,
\qquad
\mu_0:=4L_r^2\beta_{\max}.
\end{align}
At epoch $m$, the quantities used for the measurements are
\begin{align}\label{eq:qudit_alg_actions}
A^{\pm}_{m,s,i}&:=\operatorname{Retract}_{C_m}\!\left(
\pm \frac{1}{\sqrt{\lambda_{m,s-1}}}v_{m,i}
\right), \nonumber\\
O_{m,s,i}&:=\frac{1}{\sqrt{2}}
\sin\!\left(\sqrt{\frac{2}{\lambda_{m,s-1}}}\right)
v_{m,i},
\qquad
 \mathcal U_{m,s,i}^{\mathrm{tan}}(X)
    :=
    \frac{1}{2}
    \sin^2\left(\sqrt{\frac{2}{\lambda_{m,s-1}}}\right)
    \langle v_{m,i},X\rangle v_{m,i}.
\end{align}
where $\{v_{m,i}\}_{i=1}^{d_{\mathrm{tan}}}$ is an orthonormal eigenbasis of $\mathcal{V}^{\mathrm{tan}}_{m,0}$ in $T_{C_m}\mathcal{M}$. For each block $j\in[N]$, the tangent least-squares estimator at the end of epoch $m$ is
\begin{align}\label{eq:qudit_alg_lse}
\widehat{\Delta}^{(j)}_m
:=
\left(\mathcal{V}^{\mathrm{tan}}_{m,T_m}\right)^{-1}
\left( \sum_{s=1}^{T_m}\sum_{i=1}^{d_{\mathrm{tan}}}
\omega_m Y^{(j)}_{m,s,i}O_{m,s,i} \right).
\end{align}
The median-of-means aggregation is
\begin{align}\label{eq:qudit_alg_mom}
y_j
:=
\operatorname{median}_{\ell\in[N],\,\ell\neq j}
\left\|
\widehat{\Delta}^{(j)}_m-\widehat{\Delta}^{(\ell)}_m
\right\|_{\mathcal{V}^{\mathrm{tan}}_{m,T_m}},
\qquad
j_m^*:=\argmin_{j\in[N]}y_j,
\qquad
\widehat{\Delta}^{\mathrm{MoM}}_m:=\widehat{\Delta}^{(j_m^*)}_m.
\end{align}
The full algorithm can be found in Algorithm~\ref{alg:qudit_psmaqb}.

\begin{algorithm}[t]
\caption{Qudit PSMAQB}
\label{alg:qudit_psmaqb}
\footnotesize
\DontPrintSemicolon
\SetKwInput{Require}{Require}
\SetKwInput{Return}{Return}

\Require{dimension $d$, total budget $T_{\mathrm{total}}$, failure parameters $\delta,\delta_{\mathrm{w}}$, and number of epochs $M$}

Run the warm-up tomography step described above and obtain $C_1\in\mathcal{S}_d^*$ satisfying~\eqref{eq:warmup_condition_qudit} on the event $\mathcal{G}_0$\;
Set the constants as in~\eqref{eq:qudit_alg_parameters}\;

\For{$m=1,\ldots,M$}{
Set $T_m\gets\lceil\alpha\mu_{m-1}\rceil$, $\mathcal{V}^{\mathrm{tan}}_{m,0}\gets\mu_{m-1}\mathcal{I}_{T_{C_m}\mathcal{M}}$, and $\omega_m\gets\mu_{m-1}/\beta_{\mathrm{var}}$\;
Fix an orthonormal eigenbasis  $\{v_{m,i}\}_{i=1}^{d_{\mathrm{tan}}}$ of $\mathcal{V}^{\mathrm{tan}}_{m,0}$ in $T_{C_m}\mathcal{M}$\;
\For{$s=1,\ldots,T_m$}{

Construct $A^{\pm}_{m,s,i}$, $O_{m,s,i}$, and $ \mathcal U_{m,s,i}^{\mathrm{tan}}$ as in~\eqref{eq:qudit_alg_actions}, for all $i\in[d_{\mathrm{tan}}]$\;

Measure $N$ independent copies with each pair $A^+_{m,s,i},A^-_{m,s,i}$ and set $Y^{(j)}_{m,s,i}\gets (X^+_{m,s,i,j}-X^-_{m,s,i,j})/2$\;

Update $\mathcal{V}^{\mathrm{tan}}_{m,s}\gets \mathcal{V}^{\mathrm{tan}}_{m,s-1}+\omega_m\sum_{i=1}^{d_{\mathrm{tan}}} \mathcal U_{m,s,i}^{\mathrm{tan}}$\;
}

Compute $\widehat{\Delta}^{(j)}_m$ for all $j\in[N]$ using~\eqref{eq:qudit_alg_lse}\;

Compute $\widehat{\Delta}^{\mathrm{MoM}}_m$ using~\eqref{eq:qudit_alg_mom}\;

Set $\mu_m\gets\lambda_{\min}(\mathcal{V}^{\mathrm{tan}}_{m,T_m})$\;

\eIf{$\widehat{\Delta}^{\mathrm{MoM}}_m=0$}{
Set $C_{m+1}\gets C_m$\;
}{
Set $\widehat{\mathcal{V}}_m\gets\widehat{\Delta}^{\mathrm{MoM}}_m/\|\widehat{\Delta}^{\mathrm{MoM}}_m\|_F$ and
\[
\widehat{\gamma}_m\gets
\frac{1}{2}\arcsin\!\left(
\min\left\{1,\sqrt{2}\|\widehat{\Delta}^{\mathrm{MoM}}_m\|_F\right\}
\right).
\]
Set $C_{m+1}\gets\operatorname{Retract}_{C_m}(\sqrt{2}\widehat{\gamma}_m\widehat{\mathcal{V}}_m)$\;
}

Output the current estimate $C_{m+1}$\;
}

\Return{$C_{M+1}$}
\end{algorithm}

The analytic challenge of the epoch-based construction is to connect confidence bounds across tangent spaces based at different points of the manifold. At the end of epoch $m-1$, the algorithm updates the base state from $C_{m-1}$ to $C_m$. Hence the local coordinate system changes from $T_{C_{m-1}}\mathcal{M}$ to $T_{C_m}\mathcal{M}$. If all information from the previous epoch were discarded, the algorithm would repeatedly restart the exploration phase and the regret bound would become suboptimal. On the other hand, if one tried to reuse the previous observations directly, tangent vectors defined at $C_{m-1}$ would have to be projected into the new tangent space at $C_m$. This would introduce a residual term in the linear model, and this residual would accumulate into a linear contribution to the regret.
The hot start avoids both problems. At the beginning of a new epoch, we discard the raw observations and the previous tangent vectors, but we keep the statistical precision accumulated in the previous epoch. More precisely, if
\begin{align}
\mu_{m-1}:=
\lambda_{\min}\!\left(\mathcal{V}^{\mathrm{tan}}_{m-1,T_{m-1}}\right),
\end{align}
then the new tangent design superoperator is initialized as
\begin{align}
\mathcal{V}^{\mathrm{tan}}_{m,0}:=
\mu_{m-1}\mathcal{I}_{T_{C_m}\mathcal{M}}.
\end{align}
Thus, the algorithm starts epoch $m$ in the new tangent space $T_{C_m}\mathcal{M}$, but with the precision inherited from epoch $m-1$.

At the end of epoch $m-1$, the tangent-space MoM estimator gives an estimate
\begin{align}
\widehat{\Delta}^{(m-1)}_{\mathrm{MoM}}
\in T_{C_{m-1}}\mathcal{M}
\end{align}
of the local target parameter
\begin{align}
\Delta_*^{(m-1)}
=
\mathcal{P}_{T_{C_{m-1}}}(\rho-C_{m-1}).
\end{align}
The concentration bound from Section~\ref{sec:mom_qudit} controls this estimator up to the regularization bias
\begin{align}
B_{m-1}:=
\mu_{m-2}
\left(\mathcal{V}^{\mathrm{tan}}_{m-1,T_{m-1}}\right)^{-1}
(\Delta_*^{(m-1)}).
\end{align}
In particular, by the triangle inequality, with probability at least $1-\delta$,
\begin{equation}\label{eq:mom_bound_expanded}
\left\|
\widehat{\Delta}^{(m-1)}_{\mathrm{MoM}}
-
\Delta_*^{(m-1)}
\right\|_{\mathcal{V}^{\mathrm{tan}}_{m-1,T_{m-1}}}
\le
\sqrt{\beta_{\mathrm{stat}}}
+
\|B_{m-1}\|_{\mathcal{V}^{\mathrm{tan}}_{m-1,T_{m-1}}}.
\end{equation}
We define the total squared confidence radius at epoch $m-1$ as
\begin{equation}\label{eq:beta_m_definition}
\beta_{m-1}
:=
\left(
\sqrt{\beta_{\mathrm{stat}}}
+
\|B_{m-1}\|_{\mathcal{V}^{\mathrm{tan}}_{m-1,T_{m-1}}}
\right)^2 .
\end{equation}
Since
$\mathcal{V}^{\mathrm{tan}}_{m-1,T_{m-1}}
\succeq
\mu_{m-1}\mathcal{I}_{T_{C_{m-1}}\mathcal{M}}$,
the weighted confidence bound implies
\begin{equation}\label{eq:unweighted_tangent_error}
\left\|
\widehat{\Delta}^{(m-1)}_{\mathrm{MoM}}
-
\Delta_*^{(m-1)}
\right\|_F^2
\le \frac{1}{\mu_{m-1}} \left\|
\widehat{\Delta}^{(m-1)}_{\mathrm{MoM}}
-
\Delta_*^{(m-1)}
\right\|^2_{\mathcal{V}^{\mathrm{tan}}_{m-1,T_{m-1}}}
\le
\frac{\beta_{m-1}}{\mu_{m-1}}.
\end{equation}

It remains to show that the confidence radius does not grow indefinitely across epochs. We prove this in the next section. On the global success event, the update to $C_m$ stays in the region where the inverse retraction is stable, the hot start ensures the variance normalization required by the MoM estimator, and our choice of epoch length reduces the regularization bias. These ingredients imply the uniform bound
\begin{align}
\beta_m\leq \beta_{\max}
\end{align}
for all epochs.

\subsection{Uniform control of the confidence radius}

In this section we prove that the confidence radius of the qudit algorithm remains uniformly bounded over the epochs. This is the main point needed to make the hot start argument work. Indeed, at the beginning of each epoch the tangent space changes, so we need to show that the new base state is still close enough to the true state and that the next tangent-space estimator can be controlled with the same confidence radius. This is the step in Algorithm~\ref{alg:qudit_psmaqb} where, after computing $\widehat{\Delta}^{\mathrm{MoM}}_m$, we update the precision $\mu_m$ and retract from $C_m$ to $C_{m+1}$.

We will prove this by induction over the epochs. Let $\mathcal{G}_0$ be the event that the warm-up phase succeeds, namely that the initial state $C_1$ satisfies
\begin{align}
\|\rho-C_1\|_F^2\leq \frac{1}{4}.
\end{align}
For $m\geq 1$, we define $\mathcal{G}_m$ as the event that the tangent-space MoM estimator at the end of epoch $m$ satisfies the concentration bound of Theorem~\ref{thm:mom_bound}, namely
\begin{equation}\label{eq:event_Gm_qudit}
    \mathcal{G}_m
    :=
    \left\{
    \left\|
    \widehat{\Delta}^{(m)}_{\mathrm{MoM}}
    -
    \left(
    \Delta_*^{(m)}-B_m
    \right)
    \right\|^2_{\mathcal{V}^{\mathrm{tan}}_{m,T_m}}
    \leq
    \beta_{\mathrm{stat}}
    \right\}.
\end{equation}
We also define the cumulative success event
\begin{align}\label{eq:event_Em_qudit}
\mathcal{E}_m:=\bigcap_{k=0}^m \mathcal{G}_k .
\end{align}
The goal is to show that, on $\mathcal{E}_m$, the base state remains in the local region where the inverse retraction is stable and the confidence radius satisfies $\beta_m\leq \beta_{\max} $ for all epochs $m$.

We start with two elementary geometric facts. The first one relates the ambient distance between the current base state and the true state to the norm of the corresponding tangent projection. The second one shows that, if the tangent estimator is accurate enough, then the retraction update produces a new base state whose error is controlled by the tangent estimation error.

\begin{lemma} \label{lem:secant_tangent}
    For any pure states $C\in\mathcal{S}_d^*$ and $\rho\in\mathcal{S}_d^*$, let $x = \|\rho - C\|_F^2$. The  target parameter $\Delta_* = \mathcal{P}_{T_C}(\rho - C) \in T_C\mathcal{M}$ defined in~\eqref{eq:tangent_parameter} satisfies 
    \begin{equation}
        \|\Delta_*\|_F^2 = x \left( 1 - \frac{x}{2} \right).
    \end{equation}
\end{lemma}
\begin{proof}
    Parameterize $\rho = \dyad{\psi}$ and $C = \dyad{\psi_c}$, where the true state is separated from the base state by a geodesic angle $\gamma$, such that $\ket{\psi} = \cos\gamma\ket{\psi_c} + \sin\gamma\ket{\phi}$ with $\braket{\phi}{\psi_c} = 0$. The total squared error is $x = \|\rho - C\|_F^2 = 2(1-\cos^2\gamma) = 2\sin^2\gamma$. 
    The tangent projection~\eqref{eq:projector_def} evaluates to $\Delta_* = C(\rho-C)(\mathbb{I}-C) + (\mathbb{I}-C)(\rho-C)C$. Substituting the parameterization yields $\Delta_* = \cos\gamma\sin\gamma(\dyad{\psi_c}{\phi} + \dyad{\phi}{\psi_c})$. Its squared Frobenius norm evaluates to $\|\Delta_*\|_F^2 = 2\cos^2\gamma\sin^2\gamma = 2\sin^2\gamma(1-\sin^2\gamma)$. Substituting $\sin^2\gamma = x/2$  yields $\|\Delta_*\|_F^2 = 2(x/2)(1 - x/2) = x(1 - x/2)$.
\end{proof}

\begin{lemma}\label{lem:inverse_retract}
    Let $C\in\mathcal{S}_d^*$ and let $\rho\in\mathcal{S}_d^*$ be such that $
        \|\rho-C\|_F^2\leq \frac{1}{2}$.
    Let $\Delta_*:=\mathcal{P}_{T_C}(\rho-C)\in T_C\mathcal{M}$ defined in~\eqref{eq:tangent_parameter}.
    Given a tangent estimator $\hat{\Delta}\in T_C\mathcal{M}$, define $
        \hat{{V}}:=\frac{\hat{\Delta}}{\|\hat{\Delta}\|_F}$ whenever $\hat{\Delta}\neq 0$, and use the convention that $C_{\mathrm{new}}=C$ if $\hat{\Delta}=0$. We define the angle
    \begin{equation}
        \hat{\gamma}(\hat{\Delta})
        :=
        \frac{1}{2}
        \arcsin\!\left(
        \min\left\{1,\sqrt{2}\|\hat{\Delta}\|_F\right\}
        \right).
    \end{equation}
    The update of the base state is
    \begin{equation}\label{eq:anchor_update_exact}
        C_{\mathrm{new}}
        =
        \operatorname{Retract}_{C}
        \left(
        \sqrt{2}\,\hat{\gamma}(\hat{\Delta})\hat{V}
        \right).
    \end{equation}
    If $\hat{\Delta}=\Delta_*$, then $C_{\mathrm{new}}=\rho$. Moreover, for any estimator satisfying $\|\hat{\Delta}\|_F\leq \sqrt{\frac{3}{8}},$
    we have
    \begin{equation}
        \|\rho-C_{\mathrm{new}}\|_F^2
        \leq
        L_r^2
        \|\hat{\Delta}-\Delta_*\|_F^2
        \qquad
        \textrm{with} \quad L_r=6.
    \end{equation}
\end{lemma}

\begin{proof}
    We first check that the update is exact when the tangent estimator is equal to the true tangent projection. Write
    \begin{align}
        C=\dyad{\psi_c},
    \end{align}
    and choose the phase of $\ket{\psi}$ such that
    \begin{align}
    \ket{\psi}=\cos(\gamma)\ket{\psi_c}+\sin(\gamma)\ket{\phi_*},
    \end{align}
    where $\braket{\phi_*}{\psi_c}=0$. Since
     \begin{align}
        \|\rho-C\|_F^2=2\sin^2(\gamma)\leq \frac{1}{2},
    \end{align}
    we have $\gamma\leq \pi/6$. In particular, $2\gamma\leq \pi/3$, so the inverse of the sine is taken on the correct branch.

    If $\Delta_*=0$, then Lemma~\ref{lem:secant_tangent} implies $\rho=C$ under the assumption $\|\rho-C\|_F^2\leq 1/2$, and the claim is trivial. Otherwise, the true tangent projection is
     \begin{align}
        \Delta_*
        =
        \cos(\gamma)\sin(\gamma)
        \left(
        \ketbra{\phi_*}{\psi_c}+\ketbra{\psi_c}{\phi_*}
        \right)
        =
        \frac{1}{\sqrt{2}}\sin(2\gamma)V_*,
    \end{align}
    where
     \begin{align}
        V_*=\frac{1}{\sqrt{2}}
        \left(
        \ketbra{\phi_*}{\psi_c}+\ketbra{\psi_c}{\phi_*}
        \right)
    \end{align}
    is the unit tangent direction. Hence
     \begin{align}
        \|\Delta_*\|_F=\frac{1}{\sqrt{2}}\sin(2\gamma).
    \end{align}
    Therefore, when $\hat{\Delta}=\Delta_*$, the angle recovered by the algorithm is
    \begin{align}
        \hat{\gamma}
        =
        \frac{1}{2}\arcsin\!\left(\sqrt{2}\|\Delta_*\|_F\right)
        =
        \frac{1}{2}\arcsin(\sin(2\gamma))
        =
        \gamma.
    \end{align}
    Substituting the step size $\tau=\sqrt{2}\gamma$ into the retraction formula~\eqref{eq:qudit_retract} gives
     \begin{align}
        \operatorname{Retract}_{C}(\sqrt{2}\gamma V_*)
        =
        \dyad{\psi}
        =
        \rho.
    \end{align}

    It remains to prove the stability bound. Since the true tangent vector also satisfies
    \begin{align}
        \|\Delta_*\|_F^2
        =
        \|\rho-C\|_F^2
        \left(
        1-\frac{\|\rho-C\|_F^2}{2}
        \right)
        \leq
        \frac{3}{8}
    \end{align}
    by Lemma~\ref{lem:secant_tangent}, it is enough to show that the map
    \begin{align}
        \hat{\Delta}\mapsto
        \operatorname{Retract}_{C}
        \left(
        \sqrt{2}\,\hat{\gamma}(\hat{\Delta})\hat{V}
        \right)
    \end{align}
    is $L_r$-Lipschitz on the domain $\|\hat{\Delta}\|_F\leq \sqrt{3/8}$.

    We assume first that $\hat{\Delta}\neq 0$ and write
    \begin{align}
        x:=\|\hat{\Delta}\|_F.
    \end{align}
    Since $x\leq \sqrt{3/8}$, we have $
\sqrt{2}x\leq \frac{\sqrt{3}}{2}<1 $. Therefore the minimum in the definition of $\hat{\gamma}$ is attained by $\sqrt{2}x$, and $\hat{\gamma}
=
\frac{1}{2}\arcsin(\sqrt{2}x).$ Equivalently, $
\sin(2\hat{\gamma})=\sqrt{2}x.$ Define
    \begin{align}
        W:=\cos(2\hat{\gamma})=\sqrt{1-2x^2}.
    \end{align}
    The half-angle identities give
    \begin{align}
        \cos^2(\hat{\gamma})
        &=
        \frac{1+\cos(2\hat{\gamma})}{2}
        =
        \frac{1+W}{2},\\
        \sin^2(\hat{\gamma})
        &=
        \frac{1-\cos(2\hat{\gamma})}{2}
        =
        \frac{1-W}{2}.
    \end{align}
    Substituting these identities and $\hat{V}=\hat{\Delta}/x$ into the retraction formula of $C_{\mathrm{new}}$~\eqref{eq:anchor_update_exact} gives
    \begin{align}
        C_{\mathrm{new}}
        &=
        \cos^2(\hat{\gamma})C
        +
        \sin^2(\hat{\gamma})\dyad{\phi}
        +
        \frac{1}{\sqrt{2}}\sin(2\hat{\gamma})\hat{V}
        \nonumber\\
        &=
        \frac{1+W}{2}C
        +
        \frac{1-W}{2}\dyad{\phi}
        +
        \frac{1}{\sqrt{2}}(\sqrt{2}x)\frac{\hat{\Delta}}{x}
        \nonumber\\
        &=
        \frac{1+W}{2}C
        +
        \frac{1-W}{2}\dyad{\phi}
        +
        \hat{\Delta}.
    \end{align}
    Since $\hat{\Delta}\in T_C\mathcal{M}$, we can write
    \begin{align}
        \hat{\Delta}
        =
        \frac{x}{\sqrt{2}}
        \left(
        \ketbra{\phi}{\psi_c}+\ketbra{\psi_c}{\phi}
        \right) , \quad \text{therefore} \quad   \hat{\Delta}^2
        =
        \frac{x^2}{2}
        \left(
        \dyad{\phi}+C
        \right),
    \end{align}
    and hence
    \begin{align}
        \dyad{\phi}
        =
        \frac{2}{x^2}\hat{\Delta}^2-C.
    \end{align}
    Substituting this into the expression for $C_{\mathrm{new}}$ gives
    \begin{align}
        C_{\mathrm{new}}
        &=
        \frac{1+W}{2}C
        +
        \frac{1-W}{2}
        \left(
        \frac{2}{x^2}\hat{\Delta}^2-C
        \right)
        +
        \hat{\Delta}
        \nonumber\\
        &=
        \left(
        \frac{1+W}{2}
        -
        \frac{1-W}{2}
        \right)C
        +
        \frac{1-W}{x^2}\hat{\Delta}^2
        +
        \hat{\Delta}
        \nonumber\\
        &=
        WC
        +
        \frac{1-W}{x^2}\hat{\Delta}^2
        +
        \hat{\Delta}.
    \end{align}
    Using
    \begin{align}
        1-W^2=1-(1-2x^2)=2x^2,
    \end{align}
    we also have
    \begin{align}
        \frac{1-W}{x^2}
        &=
        \frac{(1-W)(1+W)}{x^2(1+W)}
        =
        \frac{1-W^2}{x^2(1+W)}
        =
        \frac{2x^2}{x^2(1+W)}
        =
        \frac{2}{1+W}.
    \end{align}
    Thus the update can be written as
    \begin{equation}\label{eq:algebraic_retract}
        C_{\mathrm{new}}
        =
        WC
        +
        \frac{2}{1+W}\hat{\Delta}^2
        +
        \hat{\Delta}.
    \end{equation}
    This formula also extends continuously to the case $\hat{\Delta}=0$.

    We now bound the derivative of this map. Let $H\in T_C\mathcal{M}$ be a tangent perturbation. Since $
        x^2=\langle \hat{\Delta},\hat{\Delta}\rangle,$  we have
    \begin{align}
        d(x^2)=2\langle \hat{\Delta},H\rangle.
    \end{align}
    Applying the chain rule to $W=(1-2x^2)^{1/2}$ gives
    \begin{align}
        dW
        &=
        \frac{1}{2\sqrt{1-2x^2}}(-2)d(x^2)
        =
        -\frac{2}{W}\langle \hat{\Delta},H\rangle,\\
        d\left(\frac{2}{1+W}\right)
        &=
        -\frac{2}{(1+W)^2}dW
        =
        \frac{4}{W(1+W)^2}
        \langle \hat{\Delta},H\rangle.
    \end{align}
    Differentiating~\eqref{eq:algebraic_retract} then gives
    \begin{align}\label{eq:dcnew}
        dC_{\mathrm{new}}
        &=
        dW\,C
        +
        d\left(\frac{2}{1+W}\right)\hat{\Delta}^2
        +
        \frac{2}{1+W}d(\hat{\Delta}^2)
        +
        d\hat{\Delta}
        \nonumber\\
        &=
        -\frac{2}{W}
        \langle \hat{\Delta},H\rangle C
        +
        \frac{4}{W(1+W)^2}
        \langle \hat{\Delta},H\rangle \hat{\Delta}^2
        +
        \frac{2}{1+W}
        (\hat{\Delta}H+H\hat{\Delta})
        +
        H.
    \end{align}

    We bound the Frobenius norm of each term. First, by Cauchy-Schwarz,
    \begin{align}
        |\langle \hat{\Delta},H\rangle|
        \leq
        \|\hat{\Delta}\|_F\|H\|_F
        =
        x\|H\|_F.
    \end{align}
    Since $C$ is a rank-$1$ projector, $\|C\|_F=1$, and the first term in~\eqref{eq:dcnew} is bounded by
    \begin{align}
        \frac{2x}{W}\|H\|_F.
    \end{align}
    For the second term, we use
    \begin{align}
        \|\hat{\Delta}^2\|_F
        =
        \left\|
        \frac{x^2}{2}(\dyad{\phi}+C)
        \right\|_F
        =
        \frac{x^2}{\sqrt{2}},
    \end{align}
    since $\dyad{\phi}$ and $C$ are orthogonal rank-$1$ projectors. Hence the quadratic term in~\eqref{eq:dcnew} is bounded by
    \begin{align}
        \frac{4}{W(1+W)^2}
        (x\|H\|_F)
        \frac{x^2}{\sqrt{2}}
        =
        \frac{2\sqrt{2}x^3}{W(1+W)^2}
        \|H\|_F.
    \end{align}
    For the third term, we use $\|AB\|_F\leq \|A\|_2\|B\|_F$. The non-zero eigenvalues of $\hat{\Delta}$ are $\pm x/\sqrt{2}$, so $\|\hat{\Delta}\|_2=\frac{x}{\sqrt{2}}$. Therefore,
    \begin{align}
        \|\hat{\Delta}H+H\hat{\Delta}\|_F
        \leq
        2\|\hat{\Delta}\|_2\|H\|_F
        =
        \sqrt{2}x\|H\|_F,
    \end{align}
    and the third term in~\eqref{eq:dcnew} is bounded by
    \begin{align}
        \frac{2\sqrt{2}x}{1+W}\|H\|_F.
    \end{align}
    The last term contributes $\|H\|_F$. Combining all bounds gives
    \begin{equation}
        \frac{\|dC_{\mathrm{new}}\|_F}{\|H\|_F}
        \leq
        L(x)
        :=
        1
        +
        \frac{2x}{W}
        +
        \frac{2\sqrt{2}x^3}{W(1+W)^2}
        +
        \frac{2\sqrt{2}x}{1+W}.
    \end{equation}
    On the interval $x\in[0,\sqrt{3/8}]$, we have $W(x)=\sqrt{1-2x^2}\geq 1/2$. A direct derivative check shows that $L(x)$ is increasing on this interval. Therefore,
    \begin{align}
        L(x)
        &\leq
        L\left(\sqrt{\frac{3}{8}}\right)
        \nonumber\\
        &=
        1
        +
        \frac{2\sqrt{3/8}}{1/2}
        +
        \frac{2\sqrt{2}(3/8)^{3/2}}{(1/2)(3/2)^2}
        +
        \frac{2\sqrt{2}\sqrt{3/8}}{3/2}
        \nonumber\\
        &\approx
        5.181
        \leq
        6.
    \end{align}
    Thus the update map is $L_r$-Lipschitz on the domain $\|\hat{\Delta}\|_F\leq \sqrt{3/8}$, with $L_r=6$. Since the map sends $\Delta_*$ to $\rho$, we finally get
    \begin{align}
        \|\rho-C_{\mathrm{new}}\|_F
        \leq
        L_r\|\hat{\Delta}-\Delta_*\|_F.
    \end{align}
    Squaring both sides gives the claimed bound.
\end{proof}

The previous lemma shows that the base-state update is stable as long as the tangent estimator stays in the local domain i.e
\begin{equation}
    \|\widehat{\Delta}\|_F\leq \sqrt{\frac{3}{8}}.
\end{equation}
Thus, before applying the inverse retraction inside the induction, we need to check that the MoM estimator produced by the algorithm indeed remains in this domain. This follows from the two inductive bounds that we want to propagate: the base state is already close to the true state, and the confidence radius has not grown beyond $\beta_{\max}$.

\begin{lemma}\label{lem:domain_invariant}
    Assume that, at the end of epoch $m-1$, the induction gives
    \begin{equation}
        \|\rho-C_{m-1}\|_F^2\leq \frac{1}{4},
        \qquad
        \beta_{m-1}\leq \beta_{\max},
    \end{equation}
    where $\beta_{m-1},\beta_{\max}$ are defined in~\eqref{eq:beta_m_definition} and \eqref{eq:qudit_alg_parameters}, respectively. Conditioned on the success event $\mathcal{G}_{m-1}$~\eqref{eq:event_Gm_qudit}, the MoM estimator satisfies
    \begin{equation}
        \big\|\widehat{\Delta}_{\mathrm{MoM}}^{(m-1)} \big
        \|_F
        <
        \sqrt{\frac{3}{8}}.
    \end{equation}
\end{lemma}

\begin{proof}
    We only need to separate the size of the estimator into the size of the true tangent projection and the estimation error. By the triangle inequality,
    \begin{equation}
        \|\widehat{\Delta}_{\mathrm{MoM}}^{(m-1)}\|_F
        \leq
        \|\Delta_*^{(m-1)}\|_F
        +
        \|\widehat{\Delta}_{\mathrm{MoM}}^{(m-1)}
        -
        \Delta_*^{(m-1)}\|_F .
    \end{equation}
    Let
    \begin{equation}
        x:=\|\rho-C_{m-1}\|_F^2 .
    \end{equation}
    By assumption, $x\leq 1/4$. Using Lemma~\ref{lem:secant_tangent}, the norm of the true tangent projection is
    \begin{align}
        \|\Delta_*^{(m-1)}\|_F^2
        =
        x\left(1-\frac{x}{2}\right) 
        \leq
        \frac{1}{4}
        \left(1-\frac{1}{8}\right) 
        =
        \frac{7}{32}.
    \end{align}
    Hence
    \begin{equation}
        \|\Delta_*^{(m-1)}\|_F
        \leq
        \sqrt{\frac{7}{32}}.
    \end{equation}

    We now bound the estimation error. On the event $\mathcal{G}_{m-1}$~\eqref{eq:event_Gm_qudit}, together with the definition of $\beta_{m-1}$~\eqref{eq:beta_m_definition} gives
    \begin{equation}
        \|\widehat{\Delta}_{\mathrm{MoM}}^{(m-1)}
        -
        \Delta_*^{(m-1)}\|^2_{\mathcal{V}^{\mathrm{tan}}_{m-1,T_{m-1}}}
        \leq
        \beta_{m-1}.
    \end{equation}
    Since the precision of epoch $m-1$ is $\mu_{m-1} =
        \lambda_{\min}\!\left(
        \mathcal{V}^{\mathrm{tan}}_{m-1,T_{m-1}}
        \right)$,
    we have $\mathcal{V}^{\mathrm{tan}}_{m-1,T_{m-1}}
        \succeq
        \mu_{m-1}\mathcal{I}_{T_{C_{m-1}}\mathcal{M}}$. Therefore,
    \begin{equation}
        \|\widehat{\Delta}_{\mathrm{MoM}}^{(m-1)}
        -
        \Delta_*^{(m-1)}\|_F^2
        \leq
        \frac{\beta_{m-1}}{\mu_{m-1}}.
    \end{equation}
    Using the inductive bound $\beta_{m-1}\leq \beta_{\max}$ and the initialization established in~\eqref{eq:qudit_alg_parameters} $\mu_{m-1}\geq \mu_0=4L_r^2\beta_{\max}$, we get
    \begin{align}
        \|\widehat{\Delta}_{\mathrm{MoM}}^{(m-1)}
        -
        \Delta_*^{(m-1)}\|_F
        &\leq
        \sqrt{
        \frac{\beta_{\max}}{4L_r^2\beta_{\max}}
        } =
        \frac{1}{2L_r}.
    \end{align}
    Combining the two bounds gives
    \begin{align}
        \|\widehat{\Delta}_{\mathrm{MoM}}^{(m-1)}\|_F
        &\leq
        \sqrt{\frac{7}{32}}
        +
        \frac{1}{2L_r}.
    \end{align}
    For the value $L_r=6$, this becomes
    \begin{align}
        \|\widehat{\Delta}_{\mathrm{MoM}}^{(m-1)}\|_F
        \leq
        \sqrt{\frac{7}{32}}
        +
        \frac{1}{12} 
        \leq
        0.56.
    \end{align}
    Since $ 0.56 < \sqrt{\frac{3}{8}}$, the MoM estimator lies in the domain required by Lemma~\ref{lem:inverse_retract}.
\end{proof}

The last lemma ensures that the retraction update is well-defined and stable under the induction hypothesis. The next step is to check that, once the new base state $C_m$ has been produced, the variance of the new difference rewards is small enough for the weighted MoM concentration bound to apply in the next epoch. This is where the particular choice
\begin{equation}
    \omega_m=\frac{\mu_{m-1}}{\beta_{\mathrm{var}}},
\end{equation}
is used.

\begin{theorem}\label{thm:variance_val}
    Fix an epoch $m$ and assume that the tangent design superoperator is hot-started as $\mathcal{V}^{\mathrm{tan}}_{m,0}
        =
        \mu_{m-1}\mathcal{I}_{T_{C_m}\mathcal{M}}$. Assume also that the new base state satisfies $
        \|\rho-C_m\|_F^2
        \leq
        \frac{L_r^2\beta_{\max}}{\mu_{m-1}}$.
    Then, for every step $s\in[T_m]$, tangent direction $i\in[d_{\mathrm{tan}}]$, and repetition
$j\in[N]$,
\begin{align}
    \omega_m
    \operatorname{Var}
    \left(
        \varepsilon^{(j)}_{m,s,i}
        \mid
        \mathcal F_{m,s-1}
    \right)
    \le 1 .
    \end{align}
\end{theorem}

\begin{proof}
    We fix a step $s$ and a tangent direction $i$ inside epoch $m$. Recall from~\eqref{eq:reward_tangent} that the difference reward is defined from the two binary outcomes as
    \begin{equation}
        Y_{m,s,i}
        =
        \frac{1}{2}
        \left(
        X^+_{m,s,i}
        -
        X^-_{m,s,i}
        \right).
    \end{equation}
    And the noise term from~\eqref{eq:mom_noise_def} is
    \begin{equation}
        \varepsilon_{m,s,i}
        =
        \frac{1}{2}
        \left(
        \eta^+_{m,s,i}
        -
        \eta^-_{m,s,i}
        \right),
    \end{equation}
    where
    \begin{equation}
        \eta^\pm_{m,s,i}
        =
        X^\pm_{m,s,i}
        -
        \operatorname{Tr}
        \left(
        \rho A^\pm_{m,s,i}
        \right).
    \end{equation}
    The two measurements are independent conditioned on the past, so
    \begin{align}
        \operatorname{Var}
        \left(
        \varepsilon_{m,s,i}\mid \mathcal{F}_{m,s-1}
        \right)
        &=
        \frac{1}{4}
        \left(
        \operatorname{Var}
        \left(
        X^+_{m,s,i}\mid \mathcal{F}_{m,s-1}
        \right)
        +
        \operatorname{Var}
        \left(
        X^-_{m,s,i}\mid \mathcal{F}_{m,s-1}
        \right)
        \right).
    \end{align}
    For a Bernoulli random variable $X$ with mean $p$, we have that  $\mathrm{Var}(X) = p(1-p)\leq 1-p$. In our case the probability is
    \begin{equation}
        p^\pm_{m,s,i}
        =
        \operatorname{Tr}
        \left(
        \rho A^\pm_{m,s,i}
        \right).
    \end{equation}
    Since both $\rho$ and $A^\pm_{m,s,i}$ are pure states,
    \begin{equation}
        1-p^\pm_{m,s,i}
        =
        \frac{1}{2}
        \|\rho-A^\pm_{m,s,i}\|_F^2 .
    \end{equation}
    Therefore combining the above we have
    \begin{align}
        \operatorname{Var}
        \left(
        \varepsilon_{m,s,i}\mid \mathcal{F}_{m,s-1}
        \right)
        &\leq
        \frac{1}{8}
        \left(
        \|\rho-A^+_{m,s,i}\|_F^2
        +
        \|\rho-A^-_{m,s,i}\|_F^2
        \right).
    \end{align}
    We now separate the distance to the action into the distance to the base state and the distance from the base state to the action. Using
    \begin{equation}
        \|X-Y\|_F^2
        \leq
        2\|X-C_m\|_F^2
        +
        2\|C_m-Y\|_F^2,
    \end{equation}
    we obtain
    \begin{align}
        \operatorname{Var}
        \left(
        \varepsilon_{m,s,i}\mid \mathcal{F}_{m,s-1}
        \right)
        &\leq
        \frac{1}{8}
        \left(
        2\|\rho-C_m\|_F^2
        +
        2\|C_m-A^+_{m,s,i}\|_F^2
        \right.\\
        &\hspace{3.5cm}
        \left.
        +
        2\|\rho-C_m\|_F^2
        +
        2\|C_m-A^-_{m,s,i}\|_F^2
        \right).
    \end{align}
    The two actions are symmetric around $C_m$, hence
    \begin{equation}
        \|C_m-A^+_{m,s,i}\|_F^2
        =
        \|C_m-A^-_{m,s,i}\|_F^2 .
    \end{equation}
    Thus
    \begin{align}
        \operatorname{Var}
        \left(
        \varepsilon_{m,s,i}\mid \mathcal{F}_{m,s-1}
        \right)
        &\leq
        \frac{1}{2}\|\rho-C_m\|_F^2
        +
        \frac{1}{2}\|C_m-A^\pm_{m,s,i}\|_F^2 .
    \end{align}
    Recall that by assumption,
    \begin{equation}
        \|\rho-C_m\|_F^2
        \leq
        \frac{L_r^2\beta_{\max}}{\mu_{m-1}} ,
    \end{equation}
    and also by construction the actions are obtained by retracting from $C_m$ with step size $
        \tau_{m,s}
        =
        \frac{1}{\sqrt{\lambda_{m,s-1}}}$,
    where
    $
        \lambda_{m,s-1}
        =
        \lambda_{\min}
        \left(
        \mathcal{V}^{\mathrm{tan}}_{m,s-1}
        \right).
    $
    Hence by our choice of action~\eqref{eq:action_from_v} together with the retraction formula~\eqref{eq:qudit_retract} and recalling that we fixed the base state $C_m$ we have
    \begin{align}
        \|C_m-A^\pm_{m,s,i}\|_F^2
        =
        2\sin^2
        \left(
        \frac{\tau_{m,s}}{\sqrt{2}}
        \right)
        \leq
        \tau_{m,s}^2
        =
        \frac{1}{\lambda_{m,s-1}} ,
    \end{align}
    where we used $\sin (x) \leq x$. Since the design superoperator is hot-started with $\mathcal{V}^{\mathrm{tan}}_{m,0}
        =
        \mu_{m-1}\mathcal{I}_{T_{C_m}\mathcal{M}}$,
    and the updates are positive semidefinite, the eigenvalue is monotone in $s$. Therefore,
    \begin{equation}
        \lambda_{m,s-1}
        \geq
        \mu_{m-1}.
    \end{equation}
    Combining the previous bounds gives
    \begin{align}
        \operatorname{Var}
        \left(
        \varepsilon_{m,s,i}\mid \mathcal{F}_{m,s-1}
        \right)
        &\leq
        \frac{L_r^2\beta_{\max}}{2\mu_{m-1}}
        +
        \frac{1}{2\lambda_{m,s-1}} .
    \end{align}
    Multiplying by the epoch weight gives
    \begin{align}
        \omega_m
        \operatorname{Var}
        \left(
        \varepsilon_{m,s,i}\mid \mathcal{F}_{m,s-1}
        \right)
        &\leq
        \frac{\mu_{m-1}}{\beta_{\mathrm{var}}}
        \left(
        \frac{L_r^2\beta_{\max}}{2\mu_{m-1}}
        +
        \frac{1}{2\lambda_{m,s-1}}
        \right)\\
        &=
        \frac{L_r^2\beta_{\max}}{2\beta_{\mathrm{var}}}
        +
        \frac{1}{2\beta_{\mathrm{var}}}
        \frac{\mu_{m-1}}{\lambda_{m,s-1}}\\
        &\leq
        \frac{L_r^2\beta_{\max}+1}{2\beta_{\mathrm{var}}}.
    \end{align}
    Using the definition $\beta_{\mathrm{var}}=L_r^2\beta_{\max}+1$, we get
        $\omega_m
        \operatorname{Var}
        \left(
        \varepsilon_{m,s,i}\mid \mathcal{F}_{m,s-1}
        \right)
        \leq
        \frac{1}{2}
        \leq
        1.$
\end{proof}

We have now proved the local ingredients needed to continue the algorithm from one epoch to the next. Lemma~\ref{lem:inverse_retract} shows that the base-state update is stable whenever the tangent estimator stays in the local domain. Lemma~\ref{lem:domain_invariant} shows that the MoM estimator indeed stays in this domain under the induction hypothesis. Theorem~\ref{thm:variance_val} shows that, once the next base state is produced, the variance normalization required by the MoM estimator remains valid in the next tangent space.

What remains is to close the induction. For this, we still need to show that the regularization bias introduced by the hot start does not accumulate across epochs. This is the only remaining term in the confidence radius $\beta_m$. The next subsection proves that the epoch schedule makes this bias contract, and then combines all the previous estimates into the final inductive statement.

\subsection{Closing the epoch induction}\label{sec:closing_epoch_induction}

We now close the induction over the epochs. From the previous subsection, we know that if the base state at the start of epoch $m$ is sufficiently close to $\rho$, then the variance normalization required by the MoM estimator is valid. It remains to control the regularization bias introduced by the hot start. The first step is to translate the ambient closeness of $C_m$ to $\rho$ into a bound on the size of the local target parameter $\Delta_*^{(m)}$ in the new tangent-space norm.

\begin{lemma}\label{lem:coordinate_transfer}
    At the start of epoch $m$, assume that $
        \|\rho-C_m\|_F^2
        \leq
        \frac{L_r^2\beta_{\max}}{\mu_{m-1}},$
    and that the tangent design superoperator is hot-started as $
        \mathcal{V}^{\mathrm{tan}}_{m,0}
        =
        \mu_{m-1}\mathcal{I}_{T_{C_m}\mathcal{M}}.$
    Then the target parameter $
        \Delta_*^{(m)}
        =
        \mathcal{P}_{T_{C_m}}(\rho-C_m)$,
    satisfies
    \begin{equation}
        \|\Delta_*^{(m)}\|_{\mathcal{V}^{\mathrm{tan}}_{m,0}}^2
        \leq
        L_r^2\beta_{\max}.
    \end{equation}
\end{lemma}

\begin{proof}
    The map $\mathcal{P}_{T_{C_m}}$ defined in~\eqref{eq:projector_def} is an orthogonal projector with respect to the Frobenius inner product by Lemma~\ref{lem:projector_properties}. Hence it is non-expansive, and using the assumption $
        \|\rho-C_m\|_F^2
        \leq
        \frac{L_r^2\beta_{\max}}{\mu_{m-1}},$  we have
    \begin{align}
        \|\Delta_*^{(m)}\|_F^2
        &=
        \|\mathcal{P}_{T_{C_m}}(\rho-C_m)\|_F^2 \\
        &\leq
        \|\rho-C_m\|_F^2 \\
        &\leq
        \frac{L_r^2\beta_{\max}}{\mu_{m-1}}.
    \end{align}
    Since the design superoperator is initialized as $\mathcal{V}_{m,0}^{\mathrm{tan}} = \mu_{m-1}\mathcal{I}_{T_{C_m}\mathcal{M}}$, we have
    \begin{align}
        \|\Delta_*^{(m)}\|_{\mathcal{V}^{\mathrm{tan}}_{m,0}}^2
        &=
        \left\langle
        \Delta_*^{(m)},
        \mathcal{V}^{\mathrm{tan}}_{m,0}(\Delta_*^{(m)})
        \right\rangle \\
        &=
        \mu_{m-1}
        \|\Delta_*^{(m)}\|_F^2 \\
        &\leq
        L_r^2\beta_{\max}.
    \end{align}
    This proves the claim.
\end{proof}

This lemma shows that the target parameter does not become large when we move to the new tangent space. Thus the hot start gives a controlled initial norm for $\Delta_*^{(m)}$. The only possible loss in the confidence radius now comes from the regularization bias
\begin{equation}
    B_m
    =
    \mu_{m-1}
    \left(
    \mathcal{V}^{\mathrm{tan}}_{m,T_m}
    \right)^{-1}
    \Delta_*^{(m)}.
\end{equation}
The next result shows that the epoch length is chosen  so that this bias contracts before the next update.
\begin{theorem}\label{thm:bias_contraction}
    Fix an epoch $m$ and assume the bound of Lemma~\ref{lem:coordinate_transfer} that is $\|\Delta_*^{(m)}\|_{\mathcal{V}^{\mathrm{tan}}_{m,0}}^2
        \leq
        L_r^2\beta_{\max}$. Let
    \begin{equation}
        B_m
        :=
        \mu_{m-1}
        \left(
        \mathcal{V}^{\mathrm{tan}}_{m,T_m}
        \right)^{-1}
        \Delta_*^{(m)}
    \end{equation}
    be the regularization bias at the end of the epoch. If $
        T_m
        =
        \lceil \alpha\mu_{m-1}\rceil,
        \alpha
        =
        \left\lceil
        \frac{8L_r^4\beta_{\mathrm{var}}}{c_0^2}
        \right\rceil, $ and $\omega = \mu_{m-1}/\beta_{\mathrm{var}}$,
    then
    \begin{equation}
        \|B_m\|^2_{\mathcal{V}^{\mathrm{tan}}_{m,T_m}}
        \leq
        \frac{\beta_{\max}}{4}.
    \end{equation}
    Consequently, the confidence radius satisfies $
        \beta_m\leq \beta_{\max}$ with $\beta_{m}$ defined as in~\eqref{eq:beta_m_definition}.
\end{theorem}

\begin{proof}
    We first bound the bias in the norm induced by the final tangent design superoperator. By definition,
    \begin{align}
        \|B_m\|^2_{\mathcal{V}^{\mathrm{tan}}_{m,T_m}}
        &=
        \mu_{m-1}^2
        \left\langle
        \left(
        \mathcal{V}^{\mathrm{tan}}_{m,T_m}
        \right)^{-1}
        \Delta_*^{(m)},
        \Delta_*^{(m)}
        \right\rangle .
    \end{align}
    Since  $
        \mu_m
        :=
        \lambda_{\min}
        \left(
        \mathcal{V}^{\mathrm{tan}}_{m,T_m}
        \right) $,
    we have
    \begin{equation}
        \left(
        \mathcal{V}^{\mathrm{tan}}_{m,T_m}
        \right)^{-1}
        \preceq
        \frac{1}{\mu_m}
        \mathcal{I}_{T_{C_m}\mathcal{M}} .
    \end{equation}
    Therefore,
    \begin{align}
        \|B_m\|^2_{\mathcal{V}^{\mathrm{tan}}_{m,T_m}}
        &\leq
        \frac{\mu_{m-1}^2}{\mu_m}
        \|\Delta_*^{(m)}\|_F^2 .
    \end{align}
    Using Lemma~\ref{lem:coordinate_transfer}, this gives
    \begin{align}
        \|B_m\|^2_{\mathcal{V}^{\mathrm{tan}}_{m,T_m}}
        \leq
        \frac{\mu_{m-1}^2}{\mu_m}
        \left(
        \frac{L_r^2\beta_{\max}}{\mu_{m-1}}
        \right)
        =
        L_r^2
        \left(
        \frac{\mu_{m-1}}{\mu_m}
        \right)
        \beta_{\max}.
    \end{align}

    It remains to lower bound the ratio $\mu_m/\mu_{m-1}$. By Theorem~\ref{thm:isotropic_growth},
    \begin{equation}
        \mu_m^2
        \geq
        \mu_{m-1}^2
        +
        2c_0^2\omega_m T_m .
    \end{equation}
    Since  $
        \omega_m
        =
        \frac{\mu_{m-1}}{\beta_{\mathrm{var}}} $,
    we get
    \begin{equation}
        \mu_m^2
        \geq
        \mu_{m-1}^2
        +
        \frac{2c_0^2}{\beta_{\mathrm{var}}}
        \mu_{m-1}T_m .
    \end{equation}
    By the choice of the epoch length,
    \begin{equation}
        T_m
        \geq
        \frac{8L_r^4\beta_{\mathrm{var}}}{c_0^2}
        \mu_{m-1}.
    \end{equation}
    Substituting this into the previous bound gives
    \begin{align}
        \mu_m^2
        &\geq
        \mu_{m-1}^2
        +
        16L_r^4\mu_{m-1}^2
        \geq
        16L_r^4\mu_{m-1}^2 .
    \end{align}
    Hence
    \begin{equation}
        \mu_m
        \geq
        4L_r^2\mu_{m-1}.
    \end{equation}
    Plugging this into the bias bound yields
    \begin{align}
        \|B_m\|^2_{\mathcal{V}^{\mathrm{tan}}_{m,T_m}}
        &\leq
        L_r^2
        \left(
        \frac{1}{4L_r^2}
        \right)
        \beta_{\max}
        =
        \frac{\beta_{\max}}{4}.
    \end{align}

    We now use the definition of the total confidence radius. Since
    \begin{equation}
        \beta_m
        =
        \left(
        \sqrt{\beta_{\mathrm{stat}}}
        +
        \|B_m\|_{\mathcal{V}^{\mathrm{tan}}_{m,T_m}}
        \right)^2,
    \end{equation}
    the previous bound gives
    \begin{align}
        \sqrt{\beta_m}
        &\leq
        \sqrt{\beta_{\mathrm{stat}}}
        +
        \frac{1}{2}\sqrt{\beta_{\max}} .
    \end{align}
    Finally, using $
        \beta_{\max}
        =
        4\beta_{\mathrm{stat}}$,
    we obtain
    \begin{align}
        \sqrt{\beta_m}
        &\leq
        \sqrt{\beta_{\mathrm{stat}}}
        +
        \frac{1}{2}
        \left(
        2\sqrt{\beta_{\mathrm{stat}}}
        \right)
        =
        2\sqrt{\beta_{\mathrm{stat}}}
        =
        \sqrt{\beta_{\max}}.
    \end{align}
    Thus $\beta_m\leq \beta_{\max}$.
\end{proof}

We can now assemble the induction step. The point is that each estimate proved above controls one possible way in which the algorithm could leave the good region. Lemma~\ref{lem:inverse_retract} and Lemma~\ref{lem:domain_invariant} control the update of the base state, Theorem~\ref{thm:variance_val} validates the variance condition needed for the next MoM estimate, and Theorem~\ref{thm:bias_contraction} prevents the confidence radius from growing across epochs. The following proposition collects these implications in the form used later in the regret analysis.

\begin{proposition}\label{prop:inductive_invariant}
    Let $m\geq 1$ and assume that the cumulative success event
    \begin{equation}
        \mathcal{E}_{m-1}
        =
        \bigcap_{k=0}^{m-1}\mathcal{G}_k
    \end{equation}
    holds where $\mathcal{G}_k$ is defined in~\eqref{eq:event_Gm_qudit}. Then, at the start of epoch $m$, the base state satisfies
    \begin{equation}
        \|\rho-C_m\|_F^2
        \leq
        \frac{L_r^2\beta_{\max}}{\mu_{m-1}}
        \leq
        \frac{1}{4}.
    \end{equation}
    Moreover, for every step $s\in[T_m]$ and every tangent direction $i\in[d_{\mathrm{tan}}]$, and repetition $j\in [N]$, the variance normalization condition
    \begin{equation}
        \omega_m
        \operatorname{Var}
        \left(
        \varepsilon_{m,s,i}
        \mid
        \mathcal{F}_{m,s-1}
        \right)
        \leq
        1
    \end{equation}
    holds. Consequently, the MoM concentration event at epoch $m$ satisfies
    \begin{equation}
        \Pr
        \left(
        \mathcal{G}_m
        \mid
        \mathcal{E}_{m-1}
        \right)
        \geq
        1-\delta.
    \end{equation}
    Finally, conditioned on $\mathcal{E}_m$, the confidence radius satisfies
    \begin{equation}
        \beta_m
        \leq
        \beta_{\max}.
    \end{equation}
\end{proposition}

\begin{proof}
    We prove the statement by induction over the epochs.

    We start with the first epoch. On the event $\mathcal{G}_0$~\eqref{eq:event_G0_qudit}, the warm-up phase gives $
        \|\rho-C_1\|_F^2
        \leq
        \frac{1}{4}.$
    By the initialization of the algorithm~\eqref{eq:qudit_alg_parameters}, $
        \mu_0
        =
        4L_r^2\beta_{\max}$.
    Hence
    \begin{equation}
        \frac{L_r^2\beta_{\max}}{\mu_0}
        =
        \frac{1}{4},
    \end{equation}
    and therefore
    \begin{equation}
        \|\rho-C_1\|_F^2
        \leq
        \frac{L_r^2\beta_{\max}}{\mu_0}
        =
        \frac{1}{4}.
    \end{equation}
    Thus, the required base-state bound holds at the start of epoch $1$.

    Since the base-state bound holds, Theorem~\ref{thm:variance_val} applies to epoch $1$. Therefore, for all steps $s\in[T_1]$ and all tangent directions $i\in[d_{\mathrm{tan}}]$,
    \begin{equation}
        \omega_1
        \operatorname{Var}
        \left(
        \varepsilon_{1,s,i}
        \mid
        \mathcal{F}_{1,s-1}
        \right)
        \leq
        1.
    \end{equation}
    This is the variance condition needed to apply the tangent-space MoM concentration bound. Thus Theorem~\ref{thm:mom_bound} gives
    \begin{equation}
        \Pr
        \left(
        \mathcal{G}_1
        \mid
        \mathcal{G}_0
        \right)
        \geq
        1-\delta.
    \end{equation}
    Finally, conditioned on $\mathcal{E}_1=\mathcal{G}_0\cap\mathcal{G}_1$, Theorem~\ref{thm:bias_contraction} gives
    \begin{equation}
        \beta_1
        \leq
        \beta_{\max}.
    \end{equation}
    This proves the induction statement for the first epoch.

    We now assume that the statement has been proved up to epoch $m-1$. In particular, on the event $\mathcal{E}_{m-1}$ we have
    \begin{equation}
        \beta_{m-1}
        \leq
        \beta_{\max} , \quad \text{and} \quad
        \|\rho-C_{m-1}\|_F^2
        \leq
        \frac{1}{4}.
    \end{equation}
    Since $\mathcal{E}_{m-1}$ contains $\mathcal{G}_{m-1}$, Lemma~\ref{lem:domain_invariant} applies to the estimator produced at the end of epoch $m-1$. Hence
    \begin{equation}
        \|\widehat{\Delta}_{\mathrm{MoM}}^{(m-1)}\|_F
        <
        \sqrt{\frac{3}{8}}.
    \end{equation}
    Therefore the inverse-retraction stability bound from Lemma~\ref{lem:inverse_retract} can be applied to the update from $C_{m-1}$ to $C_m$. We obtain
    \begin{equation}
        \|\rho-C_m\|_F^2
        \leq
        L_r^2
        \left\|
        \widehat{\Delta}_{\mathrm{MoM}}^{(m-1)}
        -
        \Delta_*^{(m-1)}
        \right\|_F^2 .
    \end{equation}
    On $\mathcal{E}_{m-1}$, the definition of $\beta_{m-1}$ and the precision $\mu_{m-1}$ give
    \begin{equation}
        \left\|
        \widehat{\Delta}_{\mathrm{MoM}}^{(m-1)}
        -
        \Delta_*^{(m-1)}
        \right\|_F^2
        \leq
        \frac{\beta_{m-1}}{\mu_{m-1}}.
    \end{equation}
    Combining the last two inequalities and using $\beta_{m-1}\leq\beta_{\max}$ yields
    \begin{align}
        \|\rho-C_m\|_F^2
        &\leq
        \frac{L_r^2\beta_{m-1}}{\mu_{m-1}} 
        \leq
        \frac{L_r^2\beta_{\max}}{\mu_{m-1}} .
    \end{align}
    Since $
        \mu_{m-1}
        \geq
        \mu_0
        =
        4L_r^2\beta_{\max}$,
    we also have
    \begin{equation}
        \|\rho-C_m\|_F^2
        \leq
        \frac{1}{4}.
    \end{equation}
    This proves the base-state bound at the start of epoch $m$.

    With this bound in hand, Theorem~\ref{thm:variance_val} applies to epoch $m$. Hence, for every step $s\in[T_m]$ and every tangent direction $i\in[d_{\mathrm{tan}}]$,
    \begin{equation}
        \omega_m
        \operatorname{Var}
        \left(
        \varepsilon_{m,s,i}
        \mid
        \mathcal{F}_{m,s-1}
        \right)
        \leq
        1.
    \end{equation}
    Therefore, on $\mathcal E_{m-1}$, the assumptions of
Theorem~\ref{thm:mom_bound} are satisfied during epoch $m$. Since
$\mathcal E_{m-1}\in\mathcal H_m$, the tower property of expectation gives
\begin{align}
\begin{aligned}
\Pr\left(
    \mathcal G_m
    \mid
    \mathcal E_{m-1}
\right)
&=
\mathbb E\left[
    \Pr\left(
        \mathcal G_m
        \mid
        \mathcal H_m
    \right)
    \,\middle|\,
    \mathcal E_{m-1}
\right]  \\
&\ge
1-\exp(-N/8).
\end{aligned}
\end{align}
By the choice
\begin{align}
N=2\left\lceil 12\log(T_{\mathrm{total}}/\delta)\right\rceil,
\end{align}
we have
\begin{align}
\exp(-N/8)
\le
\left(\frac{\delta}{T_{\mathrm{total}}}\right)^3
\le
\delta,
\end{align}
and hence
\begin{align}
\Pr\left(
    \mathcal G_m
    \mid
    \mathcal E_{m-1}
\right)
\ge
1-\delta.
\end{align}
    Finally, conditioned on $\mathcal{E}_m=\mathcal{E}_{m-1}\cap\mathcal{G}_m$, the bias contraction result of Theorem~\ref{thm:bias_contraction} gives
    \begin{equation}
        \beta_m
        \leq
        \beta_{\max}.
    \end{equation}
    This completes the induction.
\end{proof}

\subsection{Regret analysis for the qudit algorithm}\label{sec:qudit_regret_analysis}

We can now prove the regret bound for the qudit algorithm. The previous section shows that, on the global success event, the base states stay close to the unknown state, the variance normalization needed for the MoM estimator remains valid, and the confidence radius is uniformly bounded by $\beta_{\max}$. The regret analysis then follows by summing the per-round infidelity of the actions chosen inside each epoch.

Let $T_0$ denote the number of samples used in the warm-up phase. Recall that
\begin{equation}
    T_0
    =
    \mathcal{O}\!\left(d^2\log(1/\delta_{\mathrm{w}})\right).
\end{equation}

For simplicity, we state the theorem for a budget containing $M$ complete epochs, so that
\begin{equation}
    T_{\mathrm{total}}
    =
    T_0
    +
    2Nd_{\mathrm{tan}}\sum_{m=1}^{M}T_m,
\end{equation}
where the factor $2$ comes from the two symmetric rank-one projectors $A^+_{m,s,i}$ and $A^-_{m,s,i}$. If the budget stops inside an epoch, the same proof applies by replacing the last epoch length by the number of steps actually played.

\begin{theorem}\label{thm:final_regret}
    Let $d\geq 2$ and consider a $d$-dimensional PSMAQB with action set $\mathcal{A}=\mathcal{S}_d^*$ and pure-state environment $\rho\in\mathcal{S}_d^*$. Run Algorithm~\ref{alg:qudit_psmaqb} with
    \begin{equation}
        \delta
        =
        \delta_{\mathrm{w}}
        =
        \frac{1}{T_{\mathrm{total}}^2}.
    \end{equation}
    Then the expected regret satisfies
    \begin{equation}
        \mathbb{E}_{\rho}\!\left[
        \operatorname{Regret}(T_{\mathrm{total}})
        \right]
        =
        \mathcal{O}\!\left(
        d^3\log^2(T_{\mathrm{total}})
        \right).
    \end{equation}
    Moreover, let $\widehat{\rho}_t$ denote the current base-state estimate of the algorithm at time $t$. Then, for $t> T_0$, $t\in[T_{\mathrm{total}}]$,
    \begin{equation}
        \mathbb{E}_{\rho}
        \left[
        1-F(\rho,\widehat{\rho}_t)
        \right]
        =
        \mathcal{O}\!\left(
        \min\left\{
        1,
        \frac{d^3\log(T_{\mathrm{total}})}{t}
        \right\}
        \right).
    \end{equation}
\end{theorem}

\begin{proof}
    We first control the probability of the global success event. Let
    \begin{equation}
        \mathcal{E}_M
        =
        \bigcap_{m=0}^{M}\mathcal{G}_m.
    \end{equation}
    By Proposition~\ref{prop:inductive_invariant}, for every $m\geq 1$,
    \begin{equation}
        \Pr\!\left(
        \mathcal{G}_m\mid \mathcal{E}_{m-1}
        \right)
        \geq
        1-\delta.
    \end{equation}
    Also, the warm-up succeeds with probability at least $1-\delta_{\mathrm{w}}$, that is,
    \begin{equation}
        \Pr(\mathcal{G}_0)
        \geq
        1-\delta_{\mathrm{w}}.
    \end{equation}
    Therefore, by the chain rule,
    \begin{align}
        \Pr(\mathcal{E}_M)
        &=
        \Pr(\mathcal{G}_0)
        \prod_{m=1}^{M}
        \Pr\!\left(
        \mathcal{G}_m\mid \mathcal{E}_{m-1}
        \right)
        \geq
        (1-\delta_{\mathrm{w}})(1-\delta)^M .
    \end{align}
    Using Bernoulli's inequality $(1-\delta)^M \geq 1-M\delta$, and the choice
    \begin{equation}
        \delta
        =
        \delta_{\mathrm{w}}
        =
        \frac{1}{T_{\mathrm{total}}^2},
    \end{equation}
    we get
    \begin{align}
        \Pr(\mathcal{E}_M^c)
        &\leq
        \delta_{\mathrm{w}}
        +
        M\delta
        \leq
        \frac{M+1}{T_{\mathrm{total}}^2}.
    \end{align}
    Since $M=\mathcal{O}(\log T_{\mathrm{total}})$, and in particular $M+1\leq T_{\mathrm{total}}$ for the non-trivial range of parameters, we will use
    \begin{equation}\label{eq:global_failure_prob}
        \Pr(\mathcal{E}_M^c)
        \leq
        \frac{1}{T_{\mathrm{total}}}.
    \end{equation}

    We now bound the regret. The regret of the warm-up phase is at most its number of samples, since the instantaneous regret is bounded by $1$. Hence
    \begin{equation}
        \operatorname{Regret}_{\mathrm{warm}}
        \leq
        T_0
        =
        \mathcal{O}\!\left(
        d^2\log(T_{\mathrm{total}})
        \right).
    \end{equation}
    On the failure event $\mathcal{E}_M^c$, we use the trivial bound
    \begin{equation}
        \operatorname{Regret}(T_{\mathrm{total}})
        \leq
        T_{\mathrm{total}}.
    \end{equation}
    Thus, by~\eqref{eq:global_failure_prob},
    \begin{equation}
        \mathbb{E}_{\rho}
        \left[
        \operatorname{Regret}(T_{\mathrm{total}}) 
         \mathbbm{1}_{\mathcal{E}_M^c}
        \right]
        \leq
        1.
    \end{equation}

    It remains to bound the regret on $\mathcal{E}_M$. For pure states,
    \begin{equation}
        1-\operatorname{Tr}(\rho A)
        =
        \frac{1}{2}\|\rho-A\|_F^2.
    \end{equation}
    Therefore the regret accumulated during the epochs can be written as
    \begin{align}
        \operatorname{Regret}_{\mathrm{epochs}}
        &=
        \frac{1}{2}
        \sum_{m=1}^{M}
        \sum_{s=1}^{T_m}
        \sum_{j=1}^{N}
        \sum_{i=1}^{d_{\mathrm{tan}}}
        \left(
        \|\rho-A^+_{m,s,i}\|_F^2
        +
        \|\rho-A^-_{m,s,i}\|_F^2
        \right).
    \end{align}
    On $\mathcal{E}_M$, Proposition~\ref{prop:inductive_invariant} gives
    \begin{equation}
        \|\rho-C_m\|_F^2
        \leq
        \frac{L_r^2\beta_{\max}}{\mu_{m-1}}.
    \end{equation}
    Moreover, the actions are obtained by retracting (recall~\eqref{eq:qudit_retract}) from $C_m$~\eqref{eq:action_from_v} with step size
    \begin{equation}
        \tau_{m,s}
        =
        \frac{1}{\sqrt{\lambda_{m,s-1}}},
    \end{equation}
    and hence
    \begin{align}
        \|C_m-A^\pm_{m,s,i}\|_F^2
        &=
        2\sin^2
        \left(
        \frac{\tau_{m,s}}{\sqrt{2}}
        \right)
        \leq
        \tau_{m,s}^2
        =
        \frac{1}{\lambda_{m,s-1}}.
    \end{align}
    Using
    \begin{equation}
        \|\rho-A^\pm_{m,s,i}\|_F^2
        \leq
        2\|\rho-C_m\|_F^2
        +
        2\|C_m-A^\pm_{m,s,i}\|_F^2,
    \end{equation}
     the symmetry of the two actions and combining the above bounds, we obtain
    \begin{align}
        \operatorname{Regret}_{\mathrm{epochs}}
        &\leq
        Nd_{\mathrm{tan}}
        \sum_{m=1}^{M}
        \sum_{s=1}^{T_m}
        \left(
        2\|\rho-C_m\|_F^2
        +
        2\|C_m-A^\pm_{m,s,i}\|_F^2
        \right)\\
        &\leq
        2Nd_{\mathrm{tan}}
        \sum_{m=1}^{M}
        \sum_{s=1}^{T_m}
        \left(
        \frac{L_r^2\beta_{\max}}{\mu_{m-1}}
        +
        \frac{1}{\lambda_{m,s-1}}
        \right).
    \end{align}

    We bound the two terms inside the last sum separately. The first term is constant during epoch $m$. Since
    \begin{equation}
        T_m
        =
        \lceil \alpha\mu_{m-1}\rceil
        \leq
        \alpha\mu_{m-1}+1,
    \end{equation}
    and $\mu_{m-1}\geq 1$, we have
    \begin{align}
        \sum_{s=1}^{T_m}
        \frac{L_r^2\beta_{\max}}{\mu_{m-1}}
        =
        \frac{L_r^2\beta_{\max}T_m}{\mu_{m-1}}
        \leq
        L_r^2\beta_{\max}(\alpha+1).
    \end{align}

    For the second term, Theorem~\ref{thm:isotropic_growth} gives
    \begin{equation}
        \lambda_{m,s-1}
        \geq
        \sqrt{
        \mu_{m-1}^2
        +
        2c_0^2\omega_m(s-1)
        }.
    \end{equation}
    Hence
    \begin{align}
        \sum_{s=1}^{T_m}
        \frac{1}{\lambda_{m,s-1}}
        &\leq
        \frac{1}{\mu_{m-1}}
        +
        \int_0^{T_m}
        \frac{dx}{
        \sqrt{
        \mu_{m-1}^2
        +
        2c_0^2\omega_m x
        }
        }\\
        &=
        \frac{1}{\mu_{m-1}}
        +
        \frac{1}{c_0^2\omega_m}
        \left(
        \sqrt{
        \mu_{m-1}^2
        +
        2c_0^2\omega_m T_m
        }
        -
        \mu_{m-1}
        \right).
    \end{align}
    By the lower bound in Theorem~\ref{thm:isotropic_growth},
    \begin{equation}
        \sqrt{
        \mu_{m-1}^2
        +
        2c_0^2\omega_m T_m
        }
        \leq
        \mu_m.
    \end{equation}
    Using
    \begin{equation}
        \omega_m
        =
        \frac{\mu_{m-1}}{\beta_{\mathrm{var}}},
    \end{equation}
    we get
    \begin{align}
        \sum_{s=1}^{T_m}
        \frac{1}{\lambda_{m,s-1}}
        &\leq
        \frac{1}{\mu_{m-1}}
        +
        \frac{\beta_{\mathrm{var}}}{c_0^2}
        \left(
        \frac{\mu_m}{\mu_{m-1}}-1
        \right).
    \end{align}
    It remains to control the ratio $\mu_m/\mu_{m-1}$. The upper bound in Theorem~\ref{thm:isotropic_growth} gives
    \begin{equation}
        \mu_m^2
        \leq
        \mu_{m-1}^2
        +
        3\omega_mT_m.
    \end{equation}
    Since $T_m\leq \alpha\mu_{m-1}+1$, this implies
    \begin{align}
        \frac{\mu_m^2}{\mu_{m-1}^2}
        &\leq
        1
        +
        \frac{3\mu_{m-1}}{\beta_{\mathrm{var}}\mu_{m-1}^2}
        (\alpha\mu_{m-1}+1)\\
        &=
        1
        +
        \frac{3\alpha}{\beta_{\mathrm{var}}}
        +
        \frac{3}{\beta_{\mathrm{var}}\mu_{m-1}}.
    \end{align}
    Since $\mu_{m-1}\geq\mu_0$, define
    \begin{equation}
        \kappa
        :=
        \sqrt{
        1
        +
        \frac{3\alpha}{\beta_{\mathrm{var}}}
        +
        \frac{3}{\beta_{\mathrm{var}}\mu_0}
        }.
    \end{equation}
    Then
    \begin{equation}
        \frac{\mu_m}{\mu_{m-1}}
        \leq
        \kappa
    \end{equation}
    for all epochs. Hence
    \begin{align}
        \sum_{s=1}^{T_m}
        \frac{1}{\lambda_{m,s-1}}
        &\leq
        \frac{1}{\mu_0}
        +
        \frac{\beta_{\mathrm{var}}}{c_0^2}
        (\kappa-1).
    \end{align}
    The important point is that this bound does not depend on $T_m$.

    Combining the two estimates, the regret accumulated in a single epoch on $\mathcal{E}_M$ is bounded by
    \begin{align}
        2Nd_{\mathrm{tan}}
        \sum_{s=1}^{T_m}
        \left(
        \frac{L_r^2\beta_{\max}}{\mu_{m-1}}
        +
        \frac{1}{\lambda_{m,s-1}}
        \right)
        &\leq
        2Nd_{\mathrm{tan}}
        \left(
        L_r^2\beta_{\max}(\alpha+1)
        +
        \frac{1}{\mu_0}
        +
        \frac{\beta_{\mathrm{var}}}{c_0^2}
        (\kappa-1)
        \right).
    \end{align}
    Therefore
    \begin{align}
        \operatorname{Regret}_{\mathrm{epochs}}
        &\leq
        2Nd_{\mathrm{tan}}M
        \left(
        L_r^2\beta_{\max}(\alpha+1)
        +
        \frac{1}{\mu_0}
        +
        \frac{\beta_{\mathrm{var}}}{c_0^2}
        (\kappa-1)
        \right)
    \end{align}
    on the event $\mathcal{E}_M$.

    We now count the number of epochs. From the lower bound in Theorem~\ref{thm:isotropic_growth} and the choice of $T_m$, we have
    \begin{align}
        \mu_m^2
        &\geq
        \mu_{m-1}^2
        +
        16L_r^4\mu_{m-1}^2
        =
        (1+16L_r^4)\mu_{m-1}^2.
    \end{align}
    Thus
    \begin{equation}
        \mu_m
        \geq
        q\mu_{m-1},
        \qquad
        q:=
        \sqrt{1+16L_r^4}>1.
    \end{equation}
    Hence the epoch lengths grow geometrically. Since $T_M\leq T_{\mathrm{total}}$, this gives
    \begin{equation}
        M
        =
        \mathcal{O}\!\left(
        \log(T_{\mathrm{total}})
        \right).
    \end{equation}

    Finally, we use the dimensional dependence of the parameters:
    \begin{equation}
        N
        =
        \mathcal{O}\!\left(
        \log(T_{\mathrm{total}})
        \right),
        \qquad
        d_{\mathrm{tan}}
        =
        2(d-1)
        =
        \mathcal{O}(d),
    \end{equation}
    and
    \begin{equation}
        \beta_{\max}
        =
        \mathcal{O}(d),
        \qquad
        \beta_{\mathrm{var}}
        =
        \mathcal{O}(d),
        \qquad
        \alpha
        =
        \mathcal{O}(d).
    \end{equation}
    Moreover, $\kappa=\mathcal{O}(1)$ because $\alpha/\beta_{\mathrm{var}}=\mathcal{O}(1)$ and $\mu_0\geq 1$. Therefore,
    \begin{align}
        \operatorname{Regret}_{\mathrm{epochs}}
        &=
        \mathcal{O}\!\left(
        d^3\log^2(T_{\mathrm{total}})
        \right)
    \end{align}
    on $\mathcal{E}_M$. Adding the warm-up contribution and the failure-event contribution gives
    \begin{align}
        \mathbb{E}_{\rho}
        \left[
        \operatorname{Regret}(T_{\mathrm{total}})
        \right]
        &\leq
        \mathcal{O}\!\left(
        d^2\log(T_{\mathrm{total}})
        \right)
        +
        \mathcal{O}\!\left(
        d^3\log^2(T_{\mathrm{total}})
        \right)
        +
        \mathcal{O}(1)\\
        &=
        \mathcal{O}\!\left(
        d^3\log^2(T_{\mathrm{total}})
        \right).
    \end{align}

    We finish with the online estimation guarantee. During epoch $m$, the algorithm can output the current base state
    \begin{equation}
        \widehat{\rho}_t
        =
        C_m.
    \end{equation}
    On $\mathcal{E}_M$, Proposition~\ref{prop:inductive_invariant} gives
    \begin{align}
        1-F(\rho,\widehat{\rho}_t)
        &=
        1-\operatorname{Tr}(\rho C_m)
        =
        \frac{1}{2}\|\rho-C_m\|_F^2
        \leq
        \frac{L_r^2\beta_{\max}}{2\mu_{m-1}}.
    \end{align}
    We now relate $\mu_{m-1}$ to the global time $t$. Let
    \begin{equation}
        \mathcal{T}_m
        :=
        T_0
        +
        2Nd_{\mathrm{tan}}\sum_{\ell=1}^{m}T_\ell
    \end{equation}
    be the number of samples used up to the end of epoch $m$. If $t\in(\mathcal{T}_{m-1},\mathcal{T}_m]$, then $t\leq \mathcal{T}_m$. Since $T_\ell\leq(\alpha+1)\mu_{\ell-1}$ and $\mu_\ell\geq q\mu_{\ell-1}$, we have
    \begin{align}
        \sum_{\ell=1}^{m}T_\ell
        &\leq
        (\alpha+1)
        \sum_{\ell=1}^{m}
        \mu_{\ell-1}
        \leq
        (\alpha+1)
        \left(
        \frac{q}{q-1}
        \right)
        \mu_{m-1}.
    \end{align}
    Hence
    \begin{equation}
        t
        \leq
        T_0
        +
        2Nd_{\mathrm{tan}}(\alpha+1)
        \left(
        \frac{q}{q-1}
        \right)
        \mu_{m-1}.
    \end{equation}
    If $t\leq 2T_0$, then the trivial bound
    \begin{equation}
        1-F(\rho,\widehat{\rho}_t)
        \leq
        1
    \end{equation}
    is enough. If $t>2T_0$, then the previous equation implies
    \begin{equation}
        \mu_{m-1}
        \geq
        \frac{t}{
        4Nd_{\mathrm{tan}}(\alpha+1)
        }
        \left(
        \frac{q-1}{q}
        \right),
    \end{equation}
    where we absorbed the factor coming from $T_0\leq t/2$. Therefore, on $\mathcal{E}_M$,
    \begin{align}
        1-F(\rho,\widehat{\rho}_t)
        &\leq
        \frac{
        2L_r^2\beta_{\max}Nd_{\mathrm{tan}}(\alpha+1)
        }{t}
        \left(
        \frac{q}{q-1}
        \right).
    \end{align}
    Since $N=\mathcal{O}(\log T_{\mathrm{total}})$, $d_{\mathrm{tan}}=\mathcal{O}(d)$, $\alpha=\mathcal{O}(d)$ and $\beta_{\max}=\mathcal{O}(d)$, this gives
    \begin{equation}
        1-F(\rho,\widehat{\rho}_t)
        =
        \mathcal{O}\!\left(
        \frac{d^3\log(T_{\mathrm{total}})}{t}
        \right)
    \end{equation}
    on the success event, whenever $t>2T_0$. Combining this with the trivial bound for $t\leq 2T_0$ gives the bound with the minimum.

    Finally, on the failure event we use again
    \begin{equation}
        1-F(\rho,\widehat{\rho}_t)
        \leq
        1
    \end{equation}
    and~\eqref{eq:global_failure_prob}. Since $t\leq T_{\mathrm{total}}$, the failure contribution is at most $1/T_{\mathrm{total}}\leq 1/t$, which is absorbed in the same bound. Therefore,
    \begin{equation}
        \mathbb{E}_{\rho}
        \left[
        1-F(\rho,\widehat{\rho}_t)
        \right]
        =
        \mathcal{O}\!\left(
        \min\left\{
        1,
        \frac{d^3\log(T_{\mathrm{total}})}{t}
        \right\}
        \right).
    \end{equation}
\end{proof}

\section{Discussion and open problems}\label{sec:discussion}

We have shown that the low-regret tomography phenomenon is not restricted to
qubits. Although the Bloch-sphere reduction breaks down in higher dimension, the
pure-state manifold still has enough structure to support an adaptive protocol
with polylogarithmic cumulative regret. The main idea is to replace the global
linear-bandit picture by an intrinsic local one: each epoch constructs a linear
model in the tangent space of the current base state, uses symmetric retracted
measurements to cancel curvature terms, and transfers statistical precision
across changing tangent spaces through a hot-start regularization. Together with
variance-adaptive Median-of-Means regression, this gives a robust way to exploit
the vanishing-variance structure of pure-state measurements directly on
$\mathbb{CP}^{d-1}$.

The most immediate open problem is the dimensional dependence. The present
analysis gives
\begin{align}
    \mathbb E[\operatorname{Regret}(T)]
    =
    \mathcal O(d^3\log^2 T),
\end{align}
and the corresponding online infidelity guarantee scales as
\begin{align}
    \mathbb E[1-F(\rho,\widehat\rho_t)]
    =
    \mathcal O\!\left(\frac{d^3\log T}{t}\right).
\end{align}
We do not expect the factor $d^3$ to be tight. A natural goal is to tighten
the analysis, or design a sharper protocol, so that the regret matches the
optimal dimensional dependence suggested by the regret lower bound $\Omega (d\log T)$~\cite{lumbreras2024learning}, up to
logarithmic factors. Achieving this would also improve the online infidelity
bound to the optimal dimensional scaling for pure-state tomography that for single copy adaptive measurements is $\widetilde{\Theta} (d/t)$~\cite{chen_adaptive}. In this
sense, the remaining gap in the regret bound is directly tied to whether one can
obtain fully sample-optimal adaptive tomography while keeping the cumulative
disturbance essentially negligible.

A second direction is to understand the scope of the geometric tools introduced
here. Tangent-space linearization, symmetric retractions, and hot-started local
estimators are not specific to the final regret calculation; they provide a
general way to build adaptive estimators on curved quantum statistical models
while retaining linear least-squares structure locally. We expect these ideas to
be useful in other adaptive quantum learning and tomography problems where the
parameter space is a nonlinear manifold embedded in a larger linear space. This
includes, for example, adaptive tomography for structured quantum states,
learning problems with low-rank or manifold constraints.

Finally, it remains open to identify further operational interpretations and
applications of the regret objective itself. In the pure-state setting, $1-\operatorname{Tr}(\rho A_t)$
is simultaneously the instantaneous infidelity between the unknown state and the
chosen measurement direction, one half of the squared Frobenius distance, and a
quantity controlling the expected post-measurement disturbance of the consumed
copy. In state-agnostic work extraction it also appears as the cumulative
dissipated energy caused by suboptimal control directions~\cite{lumbreras25dissipation}. An interesting
question is whether the same regret objective arises in other
information-processing tasks, for instance in adaptive verification,
state-discrimination protocols, feedback control, or online calibration of
quantum devices. Such connections would further clarify when low-regret learning
is not only a statistical objective, but also the right operational notion of
gentle adaptive quantum tomography.

\paragraph*{Acknowledgements} JL is supported by the National Research Foundation through the NRF Investigatorship on Quantum-Enhanced Agents (Grant No. NRF-NRFI09-0010), the RIE25 Japan-Singapore Joint Call on Quantum R25J4IR111, the Singapore Ministry of Education Tier 1 Grant RT4/23 and RG91/25 (S) and the National Quantum Office, hosted in A*STAR, under its Centre for Quantum Technologies Funding Initiative (S24Q2d0009). MT is supported by the NRF Investigatorship award (NRF-NRFI10-2024-0006).

\bibliographystyle{ultimate}
\bibliography{biblio_purestatebandit}

@article{lumbreras22bandit,
  title = {Multi-armed quantum bandits: {E}xploration versus exploitation when learning properties of quantum states},
  author = {Lumbreras, Josep and Haapasalo, Erkka and Tomamichel, Marco},
  journal = {{Quantum}},
  publisher = {{Verein zur F{\"{o}}rderung des Open Access Publizierens in den Quantenwissenschaften}},
  volume = {6},
  pages = {749},
  year = {2022}
}

@inproceedings{heavy_tail_linear_optimal,
author = {Shao, Han and Yu, Xiaotian and King, Irwin and Lyu, Michael R.},
title = {Almost optimal algorithms for linear stochastic bandits with heavy-tailed payoffs},
year = {2018},
publisher = {Curran Associates Inc.},
address = {Red Hook, NY, USA}, 
booktitle = {Proceedings of the 32nd International Conference on Neural Information Processing Systems},
pages = {8430–8439},
numpages = {10},
location = {Montr\'{e}al, Canada},
series = {NIPS'18}
}

@article{lumbreras2024learning,
  title={Learning pure quantum states (almost) without regret},
  author={Lumbreras, Josep and Terekhov, Mikhail and Tomamichel, Marco},
  journal={arXiv preprint arXiv:2406.18370},
  year={2024}
}

@InProceedings{pmlr-v247-lumbreras24a,
  title = 	 {Linear bandits with polylogarithmic minimax regret},
  author =       {Lumbreras, Josep and Tomamichel, Marco},
  booktitle = 	 {Proceedings of Thirty Seventh Conference on Learning Theory},
  pages = 	 {3644--3682},
  year = 	 {2024},
  volume = 	 {247},
  series = 	 {Proceedings of Machine Learning Research},
  month = 	 {30 Jun--03 Jul},
  publisher =    {PMLR}
}

@INPROCEEDINGS {chen_adaptive,
author = {S. Chen and B. Huang and J. Li and A. Liu and M. Sellke},
booktitle = {2023 IEEE 64th Annual Symposium on Foundations of Computer Science (FOCS)},
title = {When Does Adaptivity Help for Quantum State Learning?},
year = {2023},
volume = {},
issn = {},
pages = {391-404},
doi = {10.1109/FOCS57990.2023.00029},
url = {https://doi.ieeecomputersociety.org/10.1109/FOCS57990.2023.00029},
publisher = {IEEE Computer Society},
address = {Los Alamitos, CA, USA},
month = {nov}
}

@inproceedings{lin1,
  title={Stochastic Linear Optimization under Bandit Feedback.},
  author={Dani, Varsha and Hayes, Thomas P and Kakade, Sham M},
  booktitle={Proceedings of the 21st Conference on Learning Theory},
  volume={2},
  pages={3},
  year={2008}
}

@article{lin2,
author = {P. Rusmevichientong and J. N. Tsitsiklis},
title = {Linearly Parameterized Bandits},
journal = {Mathematics of Operations Research},
volume = {35},
number = {2},
pages = {395-411},
year = {2010},
doi = {10.1287/moor.1100.0446}
}

@article{kueng2017low,
  title={Low rank matrix recovery from rank one measurements},
  author={Kueng, Richard and Rauhut, Holger and Terstiege, Ulrich},
  journal={Applied and Computational Harmonic Analysis},
  volume={42},
  number={1},
  pages={88--116},
  year={2017},
  publisher={Elsevier},
  doi = {https://doi.org/10.1016/j.acha.2015.07.007}  
}

@article{guctua2020fast,
  title={Fast state tomography with optimal error bounds},
  author={Gu{\c{t}}{\u{a}}, Madalin and Kahn, Jonas and Kueng, Richard and Tropp, Joel A},
  journal={Journal of Physics A: Mathematical and Theoretical},
  volume={53},
  number={20},
  pages={204001},
  year={2020},
  publisher={IOP Publishing},
  doi = {10.1088/1751-8121/ab8111}  
}

@inproceedings{lin3,
 author = {Abbasi-Yadkori, Yasin and P\'{a}l, D\'{a}vid and Szepesv\'{a}ri, Csaba},
 booktitle = {Advances in Neural Information Processing Systems},
 pages = {},
 publisher = {Curran Associates, Inc.},
 title = {Improved Algorithms for Linear Stochastic Bandits},
 volume = {24},
 year = {2011}
}

@article{lumbreras25dissipation,
  title={Quantum state-agnostic work extraction (almost) without dissipation},
  author={Lumbreras, Josep and Huang, Ruo Cheng and Hu, Yanglin and Gu, Mile and Tomamichel, Marco},
  journal={arXiv preprint arXiv:2505.09456},
  year={2025}
}

@book{bengtsson2017geometry,
  title={Geometry of quantum states: an introduction to quantum entanglement},
  author={Bengtsson, Ingemar and {\.Z}yczkowski, Karol},
  year={2017},
  publisher={Cambridge university press}
}

@book{absil2008optimization,
  title={Optimization algorithms on matrix manifolds},
  author={Absil, P-A and Mahony, Robert and Sepulchre, Rodolphe},
  year={2008},
  publisher={Princeton University Press}
}

\end{document}